\documentclass[12pt]{article}
\usepackage{graphicx,psfrag,epsf}
\usepackage{enumerate}
\usepackage{natbib}
\usepackage{url} 
\usepackage{setspace}
\usepackage{apalike}
\usepackage{amssymb}
\usepackage{amsmath}
\usepackage{amsthm}
\newtheorem{theorem}{Theorem}[section]
\newtheorem{lemma}{Lemma}[section]

\newtheorem{assumption}{Assumption}[section]

\usepackage{rotating}
\usepackage{float}
\usepackage{fancyhdr}
\usepackage{mathtools}
\usepackage{scalefnt}
\usepackage{secdot}
\usepackage[title]{appendix}
\usepackage{bm}
\usepackage{booktabs}

\newcommand{\blind}{0}

\begin{document}

\def\spacingset#1{\renewcommand{\baselinestretch}%
	{#1}\small\normalsize} \spacingset{1}


\if0\blind
{
	\title{\bf  Assessment of Regression Models with Discrete Outcomes Using  Quasi-Empirical Residual Distribution Functions}
	\author{Lu Yang \\
		School of Statistics\\ University of Minnesota}
	\date{}
	\maketitle
} \fi

\if1\blind
{
	\bigskip
	\bigskip
	\bigskip
	\begin{center}
		{\LARGE\bf Assessment of Regression Models with Discrete Outcomes Using  Quasi-Empirical Residual Distribution Functions}
	\end{center}
	\medskip
} \fi

\bigskip

\begin{abstract}
Making informed decisions about model adequacy has  been an outstanding issue for regression models with discrete outcomes. Standard assessment tools  for such outcomes (e.g. deviance residuals) often show a large discrepancy from the hypothesized pattern 
even under the true model and are  not informative  especially when data are highly discrete (e.g. binary). To fill this gap,  we propose a  quasi-empirical residual distribution function for general discrete (e.g. ordinal and count)  outcomes  that serves as an  alternative to the empirical Cox-Snell residual distribution function.
The assessment tool we propose is a principled approach and does not 
require injecting  noise into the data.  When at least one continuous covariate is available, we show  asymptotically that the proposed function  converges uniformly to the identity function under the correctly specified model, even with highly discrete  outcomes.
Through simulation studies, we demonstrate  empirically that the proposed quasi-empirical residual distribution function outperforms   commonly used residuals  for various model assessment tasks, since it  is close to the hypothesized pattern under the true model and significantly departs from this pattern under model misspecification, and is thus an effective assessment tool. 
\end{abstract}

	\noindent%
\vfill
\newpage
\spacingset{1.5}

	\section{Introduction}\label{sec:intro}
	
	Regression models   summarize researchers' and analysts' knowledge about relationships between covariates and an outcome of interest.  For example, in auto insurance applications, when modeling the number of claims from each policyholder, Poisson distributions are commonly adopted, and a set of typical covariates including drivers' age and car model is usually used by actuaries. 
	However, researchers' prior knowledge may not fully capture the patterns in the data,
	perhaps a Poisson distribution is inappropriate or other important covariates are missed.
	Given the potential pernicious consequences of model misspecification,
	in many fields, analysts are routinely tasked with demonstrating that their models  sufficiently characterize all pertinent features of the data. 
	Hence, it is of prime importance to check the adequacy of a  model's fit and  if necessary, further refine the model.

%
	
		Towards the goal of model assessment and refinement, there are two main streams
	of literature on evaluating 
		the adequacy of a model for a  dataset at hand.
	The first one is on formal goodness-of-fit tests, which produce a single number (a test statistic or the corresponding $p$-value) and yield
	conclusions with statistical confidence. The second type of approach is  an informal and typically graphical assessment,  which is the key focus of this paper.
	Residuals play a central role in 
	examining the
	agreement between a dataset and an assumed model  both formally and informally. 
One can use graphical techniques  (\citealt{ben2004quantile}) or construct overall goodness-of-fit tests using residuals (\citealt{randles1984tests}) to assess the adequacy of a model. 
Should a model be perceived as insufficient,  model diagnostics, which aim at identifying the exact causes of misspecification, may be a further task.  
The distribution of residuals often  hints at  the potential causes of misspecification.
	There are also  model checking  methods for particular model assumptions which do not rely on residuals, e.g. \cite{cook1997graphics}. 

This article centers around developing a   tool for assessing model adequacy with discrete outcomes.
	An informative assessment tool should have the following two desirable properties. First, it is crucial 
that the tool follows a known shape or pattern under the true model and deviates from this shape with model misspecification. 
This pattern of behavior is the foundation of assessment. 
Percentile-percentile (P-P) plots are commonly employed to compare the empirical distribution of residuals with the hypothesized distribution.
Second, assessment tools should be sensitive to various types of model misspecification so as to provide effective detection of model inadequacy.

Residuals, originally rooted in linear models, have been used in regression model assessment and diagnostics extensively. 
\cite{cox1968general} generalized the idea of residuals beyond normality by creating independent and identically distributed (i.i.d.) variables  homogeneous across covariates. Their framework is compelling  for continuous outcomes. For example, for continuous observations $Y_i,i=1,\ldots,n$, the uniform Cox-Snell residuals are defined as $\hat{F}_i(Y_i)$, where $\hat{F}_i$ is the fitted distribution function of $Y_i$.
Given a well-fitting model, the Cox-Snell residuals  should present a uniform trend, and otherwise lack of fit is implied.
	However,  the effectiveness of Cox-Snell residuals does not carry over to discrete outcomes, which  in general cannot be expressed as transforms of i.i.d. variables. 	Yet regression models for discrete outcomes have been long and widely applied in many areas of research, including insurance, biology, and education, among many others. 
	The focus of model assessment with discrete outcomes has been to create approximately identically distributed variables by
searching for the optimal transformations (e.g., \citealt{pierce1986residuals}).

		In practice, there are two types of established residuals commonly adopted for discrete data (\citealt{mccullagh1989generalized}). The first type is Pearson residuals which are defined according to  
		the contribution of each observation to the Pearson goodness-of-fit statistic.
The second type is  deviance residuals which depend on the contribution of each observation
to the likelihood. 
There exist  other well-known residuals, for instance Anscombe residuals  (\citealt{anscombe1961examination}), though it has been shown that Anscombe residuals and deviance residuals behave similarly (\citealt{mccullagh1989generalized}). More recently, \cite{dunn1996randomized} proposed   randomized quantile residuals  based on the idea of continuization.
Let $a_i = \lim_{y\uparrow Y_i}\hat{F}_i(y)$ and $b_i = \hat{F}_i(Y_i)$, then the randomized quantile
residual for $Y_i$ is defined as
$ r_{R}(Y_i) = \Phi^{-1} (U_i),$
where $U_i$ is a simulated uniform random variable on the interval $(a_i,b_i]$, independent of $Y_i$, and $\Phi$ is the standard normal distribution function.
 There are also residuals built for specific types of discrete data, e.g. binary (\citealt{landwehr1984graphical}) and ordinal data  (\citealt{li2012new,liu2018residuals}).

Are these established residuals informative assessment tools for discrete outcomes, in terms of the two desirable properties listed above?
Unfortunately, 
Pearson and deviance residuals of discrete outcomes often deviate dramatically from the null shape (a normal distribution) even under the true model, contrary to the first desirable property.
It has been noted that  the level of discreteness plays a pivotal  role in the behavior of residuals, so called $m$-asymptotics, in addition to the typical $n$-asymptotics (\citealt{pierce1986residuals}). Here $m$ could be the number of trials in binomial distributions, or  Poisson means, which controls the discreteness level.  With a small $m$, the outcomes concentrate on a few  values.
When $m$ is small, deviance  and Pearson residuals could have a large discrepancy with the null pattern  even under the true  model, and a large sample size $n$ does not relieve this concern. As a result, model assessments are not illuminating in this case. 
Randomized quantile residuals,  on the other hand,
 might not be sensitive to model misspecification, as a result of 
  random noise  introduced
 (see, e.g. Figure~\ref{fig:dispsmall500}).
For the same reason, moreover, this method may not be coherent with different realizations of the random noise (Figure~\ref{fig:random}).



	In this paper, we revisit Cox-Snell residuals under discreteness.
	We  construct  a quasi-empirical residual distribution function which serves as a substitute for  the empirical distribution of   Cox-Snell residuals. 
	Instead of attempting to construct  residuals themselves as is done in most existing literature, we build a surrogate for the empirical residual distribution function with discrete outcomes   based on the idea of local averaging. The proposed  function follows the presumed pattern of  empirical residual distribution functions (an identity function).
	We show  asymptotically that the   quasi-empirical residual distribution function converges to the identity function uniformly under the correctly specified model, when at least one continuous covariate is available.
	The proof of such is nontrivial due to the unique form of the proposed function and the discreteness of the outcomes.
	Under many types of misspecification including missing covariates and overdispersion, we demonstrate empirically that the proposed  quasi-empirical residual distribution function deviates significantly from the identity function.
	Hence, it can be used as a valid alternative to the empirical residual distribution function in  P-P  plots.  

The proposed assessment tool resides in the middle ground between formal  goodness-of-fit tests and  specific diagnostic   methods.   Unlike  goodness-of-fit tests which provide a single number, what we propose is a graphical tool, and its shape carries information about the severity of misspecification.
On the other hand, compared with
 the established graphical assessment tools (i.e. residual plots), 
  our methodology possesses the two desirable properties of an informative assessment tool and thus can provide reliable judgment on model adequacy. 
  Our  tool
  conducts an overall check of model adequacy. 
	If insufficiency of fit is detected and one wants to investigate the source of misspecification,  designated diagnostic tools in the literature can be subsequently applied, e.g., identifying outliers (\citealt{pregibon1981logistic}).
	
	We highlight the contributions of the proposed assessment tool as follows. First,
	the  quasi-empirical residual distribution function  converges to the identity function uniformly under a correctly specified model, even with a small $m$. This provides   theoretical justifications
	for model assessments based on the proposed tool.
	In contrast,  standard residuals including deviance residuals are normally distributed  with an error term of order at least $O_p(m^{-1/2})$ which cannot be fixed by large sample sizes. 
	Second, the  assessment tool we propose is a principled approach and 
	does not require injecting noise into the data, such as is done for randomized quantile residuals. To the best of our knowledge, this is the only  assessment approach which is not simulation-based and still guarantees the asymptotic convergence to the null shape for discrete outcomes.
	Third, we demonstrate empirically that the proposed tool outperforms other assessment tools in terms of the two key properties   described above in various settings when at least one continuous covariate is available.
	Lastly,  our tool works for general discrete outcomes including ordinary (binary) and count data.
	
	The rest of the paper is organized as follows. In Section \ref{sec:method}, we present the quasi-empirical residual distribution function and its asymptotic properties. In Section \ref{sec:simulation}, we demonstrate the usage and properties of the proposed tool in simulated examples, and  Section \ref{sec:data} contains an application of the proposed tool on an insurance dataset. Discussion and conclusions are presented in Section \ref{sec:conc}. The Appendix includes additional theoretical and simulation results. The  proofs of the theoretical results are included in the supplementary material.

	\section{Methodology}\label{sec:method}
	\subsection{Why Not Directly Apply Cox-Snell Residuals?}
	Let $Y$ be the outcome of interest.
	Denote the distribution function of $Y$ conditional on  covariates
	$\mathbf{X}=\mathbf{x}$ as $F(y\mid\mathbf{x})=P(Y\leq y\mid\mathbf{X}= \mathbf{x})$,  where $F$ belongs to a parametric family indexed by parameters $\bm\beta$. Here $\bm \beta$ can potentially relate to location, scale, and shape parameters.
	
	 Plugging $\mathbf{X}$ and $Y$ in $F$, the variable $F(Y\mid\mathbf{X})$  is known as the probability integral transform.  If $Y$  is continuous,
	for any fixed value $s \in (0, 1)$,
	\begin{align}\label{unic}
		P\left(F(Y\mid\mathbf{X}=\mathbf{x})\leq s\right)=s,
	\end{align}
	and 
	taking expectation with respect to $\mathbf{X}$ yields 
	\begin{align}\label{uni}
		P\left(F(Y\mid\mathbf{X})\leq s\right)=s,
	\end{align}
	i.e.
	$F\left(Y\mid \mathbf{X}\right)$ is uniformly distributed. 
	Given an i.i.d. sample  $(\mathbf{X}_{i}^\top,Y_{i}), i = 1, \ldots , n$, with a fitted distribution function $\hat{F}$ associated with fitted parameters $\hat{ {\bm\beta}}$, 
	one can obtain a sequence of Cox-Snell residuals $\hat{F}\left(Y_{i}\mid \mathbf{X}_{i}\right),  i = 1, \ldots , n $ and 
	the corresponding empirical residual  distribution function
	\begin{align}\label{eq:cox}
		\hat{U}_C(s;\hat{\bm\beta})=\frac{1}{n} \sum_{i=1}^{n}1(\hat{F}\left(Y_{i}\mid \mathbf{X}_{i}\right) \leq s).
	\end{align}
	Under a correctly specified model, $\hat{U}_C(\cdot;\hat{\bm\beta})$ should be  approximately an identity  function, and  an otherwise large discrepancy indicates misspecification. 
	
	Owing to this property, it is common practice to use P-P plots to visualize the  comparison between the  empirical distribution of  Cox-Snell residuals and their null distribution function  under the true model, i.e., the identity function. Figure~\ref{fig:pp} portrays the P-P plots of the Cox-Snell residuals in simulated examples.
	In the left panel, the data are generated with a gamma regression model, and the Cox-Snell residuals are calculated using the correctly specified model. As anticipated, the Cox-Snell residuals appear to be uniform.

	\begin{figure}[!h]\centering
		\includegraphics[width=.55\textwidth]{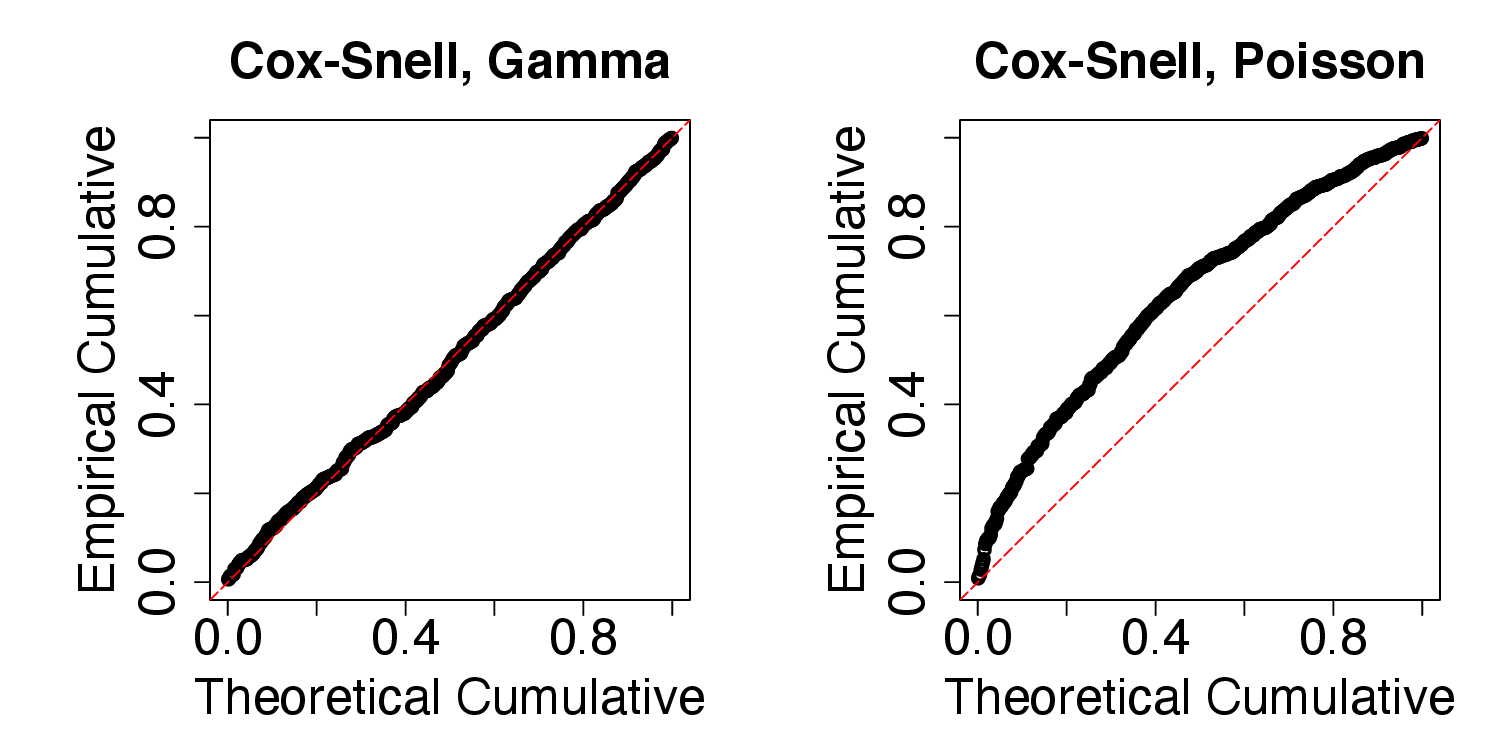}
		\caption{P-P plots of Cox-Snell residuals under correctly specified models. Left panel: data are generated and fit with gamma regression. Right panel: data are generated and fit with Poisson regression.\label{fig:pp}}
	\end{figure}
	
	However, when $Y$ is a discrete variable,
	the Cox-Snell residuals may not be uniformly distributed even under the true model. 
	As in the right panel of Figure~\ref{fig:pp},  when the data are generated from a Poisson generalized linear model (GLM), the Cox-Snell residuals are far apart from uniformity even with the knowledge of the underlying model.
	The uniformity of probability integral transforms does not hold under discreteness due to the fact that 
	\eqref{unic} is not  true for some values of $s\in(0,1)$, in contrast to the continuous cases. 
	The lemma below gives the condition under which \eqref{unic} holds for discrete outcomes.

	\begin{lemma}\label{iff}
		Without loss of generality, suppose Y is a discrete variable taking integer values. Conditioning on $\mathbf{X}=\mathbf{x}$, \eqref{unic} holds for discrete $Y$ if and only if $s=F\left(k\mid \mathbf{x}\right)$ for some integer $k$.
	\end{lemma}
\begin{proof}
	``If." Assume $s=F\left(k\mid \mathbf{x}\right)$, then $$P(F(Y\mid \mathbf{x})\leq s)=P\left(F(Y\mid \mathbf{x})\leq F\left(k\mid \mathbf{x}\right)\right)=P(Y< k+1\mid \mathbf{x})=P(Y\leq k\mid \mathbf{x})=s.$$
	
	``Only if." 
	Now suppose $P(F(Y\mid \mathbf{x})\leq s)=s$. If $s\neq F\left(k\mid \mathbf{x}\right)$ for any integer $k$,  there exists a $k_0$ such that $F(k_0\mid \mathbf{x})<s<F(k_0+1\mid \mathbf{x})$, 
	as demonstrated in Figure~\ref{fig:step}. Then $$P(F(Y\mid \mathbf{x})\leq s)=P(Y<k_0+1\mid \mathbf{x})=P(Y\leq k_0\mid \mathbf{x})=F(k_0\mid \mathbf{x})<s,$$ which is contradictory. Therefore, it holds that $s=F\left(k\mid \mathbf{x}\right)$ for some integer $k$.
	
	\begin{figure}[!h]\centering
		\includegraphics[width=.7\textwidth]{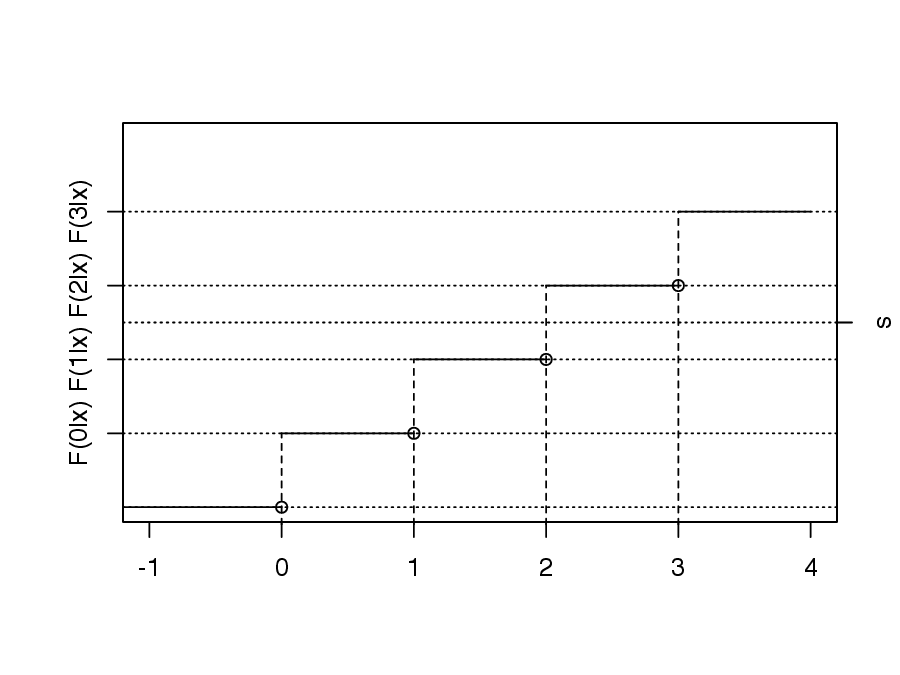}
		\vspace{-0.4in}
		\caption{Demonstrative plot for  the distribution function of a discrete random variable.\label{fig:step}}
	\end{figure}
\end{proof}

	\subsection{Construction of Quasi-Empirical Residual Distribution Functions}\label{sec:con}
	Now we construct an alternative to $\hat{U}_C(\cdot;\bm{\beta})$ under discreteness.
	Intuitively, for each point $s$, if one could find a subset of observations for which \eqref{uni} holds,  this subset can be plugged in \eqref{eq:cox} to obtain the identity function. 
	Without loss of generality,  assume hereinafter the support of $Y$ is nonnegative integers. Motivated by Lemma \ref{iff},   we  define the conditional range of the  distribution function given $\mathbf{X}=\mathbf{x}$ as a grid $\Lambda(\mathbf{x})=\{F\left(k\mid \mathbf{x}\right):k=0,1,\ldots\}$.   Note that the  range of $Y$ can be finite, for instance,  $\Lambda(\mathbf{x})=\{F\left(0\mid\mathbf{x}\right),1\}$ for binary variables.
	From Lemma \ref{iff},  \eqref{unic} is true if and only if
	$s\in \Lambda(\mathbf{x})$. 
	If $\mathbf{X}$ contains continuous components,  when $\mathbf{X}$  varies in regression,
	there might be a subset of observations for which
	$s\in \Lambda(\mathbf{X})$ holds approximately.  Equivalently, the distance between $s$ and $\Lambda(\mathbf{X})$, $d(s,\Lambda(\mathbf{X}))=\min\{\mid s-\eta\mid,\eta\in\Lambda(\mathbf{X})\}$, is small in this case.
	Hence, we need to carefully characterize  $d\left(s,\Lambda(\mathbf{X})\right)$.

	Conditioning on $\mathbf{X}=\mathbf{x}$, denote $F ^{(-1)}\left(\cdot\mid \mathbf{x}\right)$ as the general inverse function of $F\left(\cdot\mid \mathbf{x}\right)$ such that $F ^{(-1)}(s\mid  \mathbf{x})= \inf\{y :s\leq F\left(y\mid \mathbf{x}\right) <  1\} $ for $s \in (0, 1)$. 
	Here we exclude $\{1\}$ to avoid boundary effects. Removing this point is not a concern, because \eqref{unic} always holds on the boundary.
	Denote $$H^+( s;\mathbf{x}) = F (F ^{(-1)}(s\mid  \mathbf{x})\mid \mathbf{x}).$$ 
	It can be seen that $H^+( s;\mathbf{x})=\min\{\eta\in \Lambda(\mathbf{x})\backslash\{1\}: \eta\geq s\}$, i.e., $H^{+}( s;\mathbf{x})$ is the smallest interior point on the grid $\Lambda(\mathbf{x})$ that is larger than or equal to $s$.
	In the same way, one can define the largest interior point on the grid $\Lambda(\mathbf{x})$ that is smaller than or equal to $s$ as
	$$H^{-}( s;\mathbf{x})=\max\{\eta\in \Lambda(\mathbf{x})\backslash\{1\}: \eta\leq s\}.$$
	To combine these two cases, define the interior grid point closest to $s$
	\begin{align*} 
		H( s;\mathbf{x})=\begin{cases}
			H^{+}( s;\mathbf{x})&H^{+}( s;\mathbf{x})+H^{-}( s;\mathbf{x})\leq2 s\text{ or } s<F(0\mid \mathbf{x}),\\
			H^{-}( s;\mathbf{x})&H^{+}( s;\mathbf{x})+H^{-}( s;\mathbf{x})>2 s\text{ or } s>\max\left(\Lambda(\mathbf{x})\backslash\{1\}\right).
		\end{cases}
	\end{align*}
	One can view $H(s;\mathbf{x})$  as the proximal interpolator   which maps $s$ to its nearest neighbor on $\Lambda(\mathbf{x})$.
	It follows that $d(s,\Lambda(\mathbf{x})\backslash\{1\})=\mid H( s;\mathbf{x})-s\mid $.
	
	When $s$ is ``close to" being on the grid given $\mathbf{x}$ in the sense that $H(s; \mathbf{x}) \approx s$, 
	we  have an approximation to \eqref{unic}
	$$
	P\left(F(Y\mid \mathbf{x}) \le H(s; \mathbf{x})\right) =H(s; \mathbf{x}) \approx s,
	$$
	where the first equation  holds due to the fact that $H( s;\mathbf{x})\in \Lambda(\mathbf{x})$ by its definition and Lemma \ref{iff}.   Therefore, we can focus on the empirical residual distribution function among the subset of observations for which $H(s; \mathbf{X}) \approx s$.
	
	Now consider a sample $(\mathbf{X}_i^\top,Y_i),i=1,\ldots,n$ and a fixed value of $s$.  
	To realize the above idea, we use a kernel function $K(\cdot)$ to select the subset of observations whose grid is close to $s$.
	We assign weights to observations depending on the   normalized distance between $s$ and $H(s; \mathbf{X}_i)$, i.e., $K\left[(H(s; \mathbf{X}_i)-s)/\epsilon_n\right]$, where  $\epsilon_n$ is a small bandwidth. 
	Then, we focus on the empirical distribution function of the Cox-Snell residuals using the selected subset of data  and define the quasi-empirical residual distribution function  
	\begin{align}
		\begin{split}\label{esti}
			\hat{U}(s;\bm{\beta})=\sum_{i=1}^{n}W_{n}(s;\mathbf{X}_i,\bm{\beta})1\left[ F(Y_{i}\mid \mathbf{X}_i)\leq H( s;\mathbf{X}_{i})\right],
		\end{split}
	\end{align}
	where 
	$$W_{n}(s;\mathbf{X}_i,\bm{\beta})=\frac{K\left[(H( s;\mathbf{X}_{i})-s)/\epsilon_n\right]}
	{\sum_{j=1}^{n}K\left[(H( s;\mathbf{X}_{j})-s)/\epsilon_n\right]},$$
	and $K$ is a bounded, symmetric, and Lipschitz continuous kernel.

	 \sloppy
	Essentially, $\hat{U}(s;{\bm\beta})$ is an estimator of the probability
	${P}\left(Y\leq k\mid F\left(k\mid \mathbf{X}\right)=s\right)= P(F\left(Y\mid \mathbf{X}\right)\leq s \mid F\left(k\mid \mathbf{X}\right)=s)$,
	if there exists such an integer $k$ satisfying the condition $F\left(k\mid \mathbf{X}\right)=s$. 
	This probability equals  $s$ under the true model,  trivially from Lemma \ref{iff}. 
	However, it is possible that, for a given observation, there does not exist  such an integer. To overcome this problem,  we   use the subset of the data for which $F\left(k\mid \mathbf{X}\right)\approx s$ to estimate  ${P}(Y\leq k\mid F\left(k\mid \mathbf{X}\right)\approx s)$ instead. Given continuous covariates, it is possible to find such a subset,  and  kernels are employed to facilitate the selection.
	The proposed $\hat{U}(\cdot;\bm{\beta})$ should be close to the identity function under the true model, and  discrepancies from identity indicate lack of fit. 
	Hence, $\hat{U}(\cdot;\bm{\beta})$ can be used as an assessment tool in place of the empirical Cox-Snell residual distribution function  for discrete outcomes.
	For this reason, $\hat{U}(\cdot;\bm{\beta})$ is named as ``quasi"-empirical residual distribution function, implying it is an alternative to, though not the real,  empirical residual distribution function. 
	The proposed approach  does not inject noise to the data. 
	In contrast, the strategy in  \cite{dunn1996randomized} is to  simulate a uniform random variable to fill in the gap between $F\left(k\mid \mathbf{X}\right)$ and $F(k-1\mid\mathbf{X})$.

	To demonstrate the effectiveness of the proposed tool  immediately, Figure~\ref{fig:ppcomp}  shows the curve of $\hat{U}(\cdot;\bm{\beta})$ for the foregoing Poisson example  (right panel of Figure \ref{fig:pp}) wherein the model is correctly specified.
	A more thorough simulation study is included in Section \ref{sec:simulation}. 
	In practice, $\bm \beta$ is unknown;  let $\hat{\bm\beta} $ be the  corresponding estimator. 
	By plugging $\hat{\bm\beta} $ in \eqref{esti}, we may obtain  $\hat{U}(\cdot;\hat{\bm\beta})$.
	We will show in Section \ref{asym} that the uncertainty in the coefficients is negligible asymptotically.
	
\begin{figure}[!h]
	\centering
	\includegraphics[width=.3\textwidth]{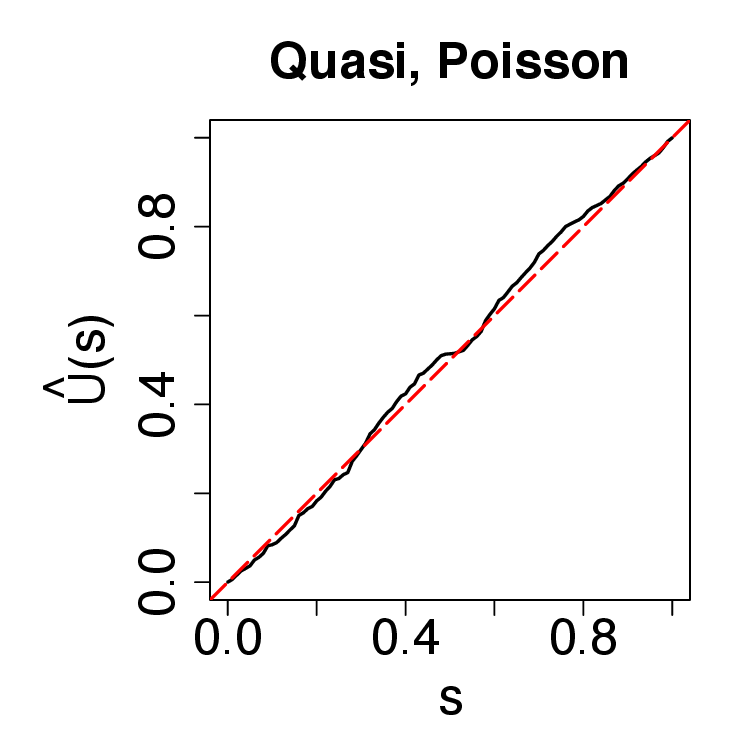} 
	
	\caption{ Curve of the proposed quasi-empirical residual distribution function for the Poisson example in the right panel of Figure \ref{fig:pp} under the correctly specified model.\label{fig:ppcomp}}
\end{figure}

	\subsection{Examples}\label{sec:exam}
	As an example, 
	when $Y$ is a binary outcome, 
	its  distribution grid only contains two points, i.e., $\Lambda(\mathbf{x})=\left\lbrace F(0\mid \mathbf{x}),1\right\rbrace$ and thus $H(s;\mathbf{x})=F(0\mid \mathbf{x})$.  Then, \eqref{esti} becomes $$\frac{\sum_{i=1}^{n}K\left[\left(F(0\mid \mathbf{X}_i)-s\right)/\epsilon_n\right]1(Y_{i}=0)}{\sum_{j=1}^{n}K\left[\left(F(0\mid \mathbf{X}_j)-s\right)/\epsilon_n\right]}.$$
	In the binary case, the proposed function takes the form of the Nadaraya-Watson estimator.
	In contrast, when $Y$ is continuous, which can be viewed as the limiting case when $Y$ gets less discrete, it is always true that $H(s;\mathbf{X}_i)=s$. It follows that $W_{n}(s;\mathbf{X}_i,\bm{\beta})=1/n$, and the quasi-empirical residual distribution function \eqref{esti} degenerates to the empirical residual distribution function \eqref{eq:cox}. In both extreme cases, the theoretical properties of the resulting quasi-empirical residual distribution function has been extensively studied  (e.g. \citealt{li2007nonparametric}).

	Technical difficulties and major departures from existing methods are pronounced when  $Y$ is discrete with an infinite range, for instance a Poisson variable, under which $\hat{U}(\cdot;\bm{\beta})$  is  nonstandard  in the following aspects. First,  for a fixed point $s$,  $H(s;\mathbf{X})$
	is a noncontinuous transform of  the random location parameter  $\mu=\mathbf{X}^\top\bm{\beta}$.
	We include an example for illustration  assuming that $Y$ follows  the commonly used   Poisson GLM  with the log link,
	i.e., $Y\mid \mathbf{X}\sim\mathrm{Poisson}\left(\mathrm{exp}(\mathbf{X}^\top\bm{\beta})\right)$.
	Figure~\ref{fkxplotapp} shows $H( s;\mathbf{X})$ (solid curves) for  fixed $s$   as a function of   $\mu$. In this example, for a fixed $k$, $F\left(k\mid \mathbf{X}\right)$ (dashed lines) is a monotone decreasing function of $\mu$. The curve of $H( s;\mathbf{X})$ as a function of $\mu$ 
	is comprised of continuous pieces from the curves of $F\left(k\mid \mathbf{X}\right), k=0,1,\ldots$, and the transitions occur when the
	nearest neighbor of $s$ changes from $F\left(k\mid \mathbf{X}\right)$ to $F(k+1\mid \mathbf{X})$ for integer $k$, as $\mu$ increases. 
	The random variable $H( s;\mathbf{X})$  is a continuous function of $\mu$ almost everywhere except at a countable number of points  under which there are two nearest neighbors of $s$ on $\Lambda(\mathbf{X})$.
	We will further address the issue of discontinuity in Section \ref{asym}, which complicates the proof of asymptotic properties.

	\begin{figure}[!h]\centering\centering
		\includegraphics[width=.75\textwidth]{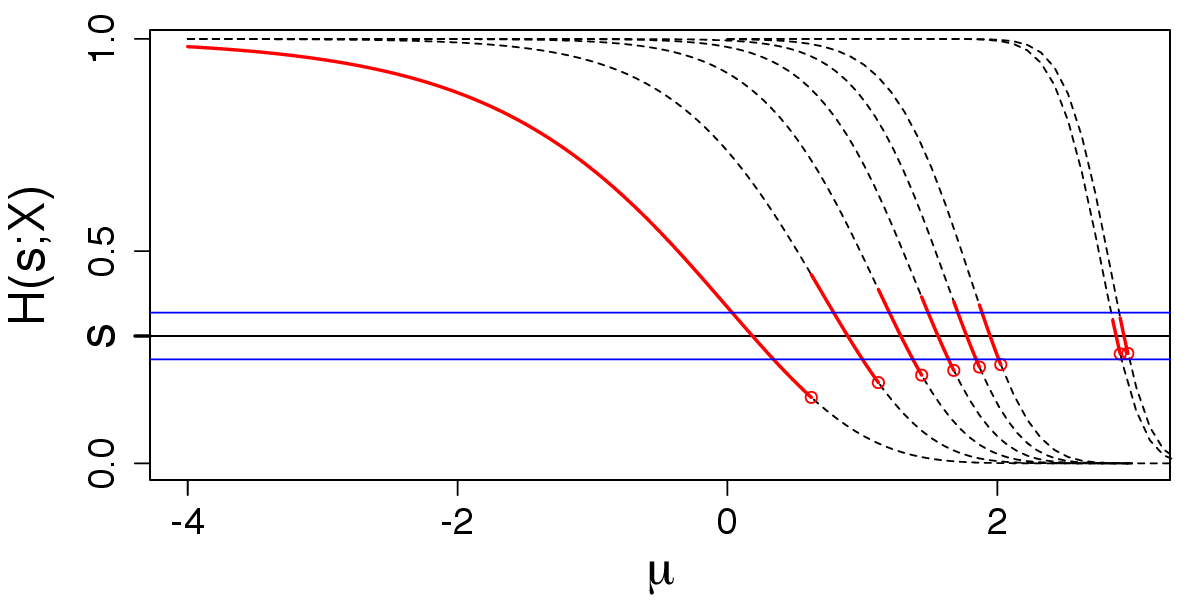}
		\caption{$H( s;\mathbf{X})$  (solid  curve) for fixed $s$ as a function of $\mu=\mathbf{X}^\top\bm{\beta}$ in a Poisson GLM with a log link. Dashed  curves: $F\left(k\mid \mathbf{X}\right)$ , from left to right $k=0,1,2,3,4,5,15,16$. The curve of $H( s;\mathbf{X})$  
			is comprised of pieces from the curves of $F\left(k\mid \mathbf{X}\right), k=0,1,\ldots$.   Horizontal lines: $s+\epsilon,s,s-\epsilon$.\label{fkxplotapp}}
	\end{figure}

	The second complicating factor is that $H( s;\mathbf{X})$, the proximal interpolator of $s$, 
	certainly changes with $s$.
	It implies that, when focusing on  different points, we plug different variables into the kernel function, which distinguishes the proposed function from traditional nonparametric regression estimators. 
	Assuming continuity of $\mu$, $F\left(k\mid \mathbf{X}\right)$  is random with a density denoted as  $f_{F\left(k\mid \mathbf{X}\right)}$, and $f_{H( s;\mathbf{X})}$ is the density of $H( s;\mathbf{X})$.  The weights in  \eqref{esti}, $W_{n}(s;\mathbf{X}_i,\bm{\beta})$,  relate to the density of   $f_{H( s;\mathbf{X})}$  at $s$, i.e.,  $f_{H( s;\mathbf{X})}(s)$. 
	By transformation of random variables,
	\begin{gather}\label{equ:hdensity}
		f_{H( s;\mathbf{X})}(s)=\sum_{k=0}^{\infty}f_{F\left(k\mid \mathbf{X}\right)}(s)\triangleq g(s).
	\end{gather}
	Note that $f_{H(s;\mathbf{X})}(t)$ is a density function with respect to $t$, 
	while  $g(s)$ is not a density function with respect to $s$.
	According to \eqref{equ:hdensity}, unlike typical nonparametric regression methods for which one assigns weights to observations depending on one regressor, here all the $F\left(k\mid \mathbf{X}\right),k=0,1,\ldots$ contribute to the weights, and this dynamic scheme increases efficiency. 
	Meanwhile, in Section \ref{asym}, by making realistic assumptions, we make sure $g(s)$ is bounded. 
For ordinal outcomes with a finite support, 
the summation in \eqref{equ:hdensity} is up to the second largest possible value. For example,
in Figure~\ref{fkxplotfinite}, $f_{H( s;\mathbf{X})}(s)=f_{F(0\mid \mathbf{X})}(s)$ in the left panel and $f_{H( s;\mathbf{X})}(s)=\sum_{k=0}^{2}f_{F\left(k\mid \mathbf{X}\right)}(s)$ in the right panel.
	\begin{figure}[!h]\centering
	\begin{center}
		\includegraphics[width=.4\textwidth]{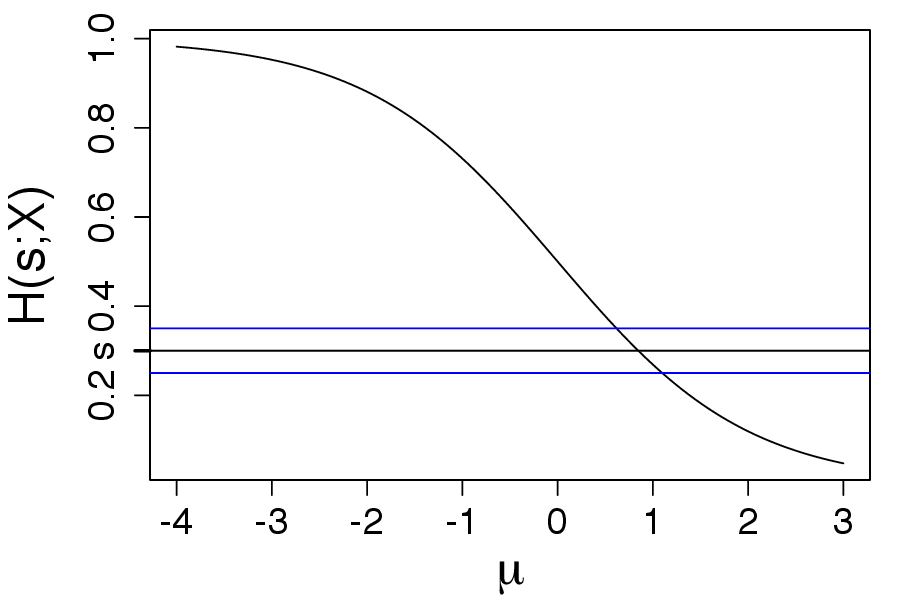}
		\includegraphics[width=.4\textwidth]{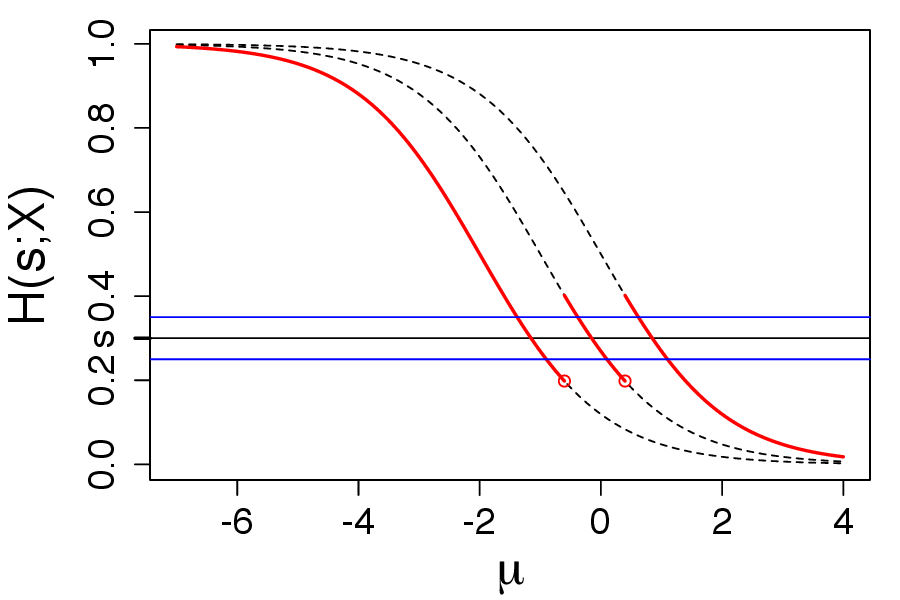}
		\caption{$H( s;\mathbf{X})$ (solid curve) as a function of $\mu=\mathbf{X}^\top\bm{\beta}$ for logistic regression (left panel) and ordinal regression with levels $0, 1, 2, 3$ (right panel). Dashed  curves for right panel: $F\left(k\mid \mathbf{X}\right)$, from left to right $k=0,1,2$.  Horizontal lines: $s+\epsilon, s,s-\epsilon$.\label{fkxplotfinite}}
	\end{center}
\end{figure}
	
	\subsection{Bandwidth Selection}\label{sec:band}
	An important choice left to be made is the bandwidth $\epsilon_n$.
	In nonparametric regression settings, there is a large body of literature on  bandwidth selection, and cross-validation approaches are routinely employed.
	However, in our application,  as discussed in Section \ref{sec:exam}, $\hat{U}(\cdot;\bm{\beta})$ is not a standard nonparametric regression estimator,  in particular when $Y$ has an infinite support. In this section, we provide a pragmatic bandwidth selector based on an approximation to the quasi-empirical residual distribution function.
	
	Generalizing the example of binary data, for each fixed integer $k$, define
	$$\hat{U}_k(s;\bm{\beta})=\sum_{i=1}^n\left[\frac{K\left[\left(F\left(k\mid \mathbf{X}_i\right)-s\right)/\epsilon_n\right]}{\sum_{j=1}^n{K\left[\left(F\left(k\mid \mathbf{X}_j\right)-s\right)/\epsilon_n\right]}}1\left( Y_{i}\leq k\right)\right].$$
	Here $\hat{U}_k(s;\bm{\beta})$ can be recognized as a Nadaraya–Watson estimator for the function $m$ such that	$P\left(Y\leq k\mid F\left(k\mid \mathbf{X}\right)\right)=m\left(F\left(k\mid \mathbf{X}\right)\right)$, which is apparently the identity function under the true model due to the fact that $P\left(Y_i\leq k\mid F\left(k\mid \mathbf{X}_i\right)=s\right)=s$, as discussed in Section \ref{sec:con}.
	From the observation of \eqref{equ:hdensity},
	all the $F\left(k\mid \mathbf{X}\right),k=0,1,\ldots$ contributes to $\hat{U}(\cdot;\bm{\beta})$, and
	the main body  of $\hat{U}(\cdot;\bm{\beta})$ can be approximated by stacking  a sequence of  
	$\hat{U}_k(s;\bm{\beta})$, i.e.
	$$\hat{U}(s;\bm{\beta})\approx \sum_{i=1}^n\sum _{ k}
	\left[\frac{K\left[\left(F\left(k\mid \mathbf{X}_i\right)-s\right)/\epsilon_n\right]}{\sum_{j=1}^n\sum _{k}{K\left[\left(F\left(k\mid \mathbf{X}_j\right)-s\right)/\epsilon_n\right]}M(k,\mathbf{X}_j)}1\left( Y_{i}\leq k\right)M(k,\mathbf{X}_i)\right],$$
	where 
	$M(k,\mathbf{X}_i)$ is a weight function such that  $0.1\leq F\left(k\mid \mathbf{X}_i\right)\leq 0.9$ in order to trim out boundary observations. 
	The above approximation   is essentially a nonparametric regression estimator for the composite data  $\left(F\left(k\mid \mathbf{X}_i\right),1(Y_i\leq k)\right),  i=1,\ldots,n, ~0.1\leq F\left(k\mid \mathbf{X}_i\right)\leq 0.9$. 
	We can thereby choose $\epsilon_n$ by leveraging   leave-one-out cross-validation and   the established bandwidth selection algorithms for nonparametric regression   (e.g. \citealt{racine2004nonparametric}). 
	A simulation study which  demonstrates the sensitivity of $\hat{U}(\cdot;\bm{\beta})$  to different bandwidths and the effectiveness of the proposed bandwidth selection rule will be  presented in Appendix \ref{sec:addsim}.
	The choice of bandwidth plays a more important role when data are highly discrete.

	\subsection{Asymptotics}\label{asym}

	In this section, we study  the asymptotic properties of the quasi-empirical residual distribution function $\hat{U}(\cdot;\hat{ {\bm\beta}})$ defined  in \eqref{esti}. 
	For the sake of simplicity, here we  show the asymptotic properties when regression is conducted on a single location parameter, though our methodology is applicable to more general settings, for instance, double GLMs wherein the dispersion parameter is a function of covariates as well, or zero-inflated Poisson models with more than one location parameter. We will demonstrate the usage of our methodology in more general situations through numerical examples.
	
	Let $\mu=\mathbf{X}^\top\bm{\beta}$ be the random location parameter. 
	For example, in Poisson GLMs, the mean  is related to the location parameter through a link function $\lambda=\exp(\mu)$, and in logistic regression the probability of one is $1/(1+\exp(-\mu))$. 
	Under an ordinal logistic regression model with proportionality, $F\left(k\mid \mathbf{X}\right)=G(\alpha_k\mid\mu),$  where  $\alpha_k$ is the threshold, and $G(\cdot\mid\mu)$ is the distribution function of a  logistic or normal random variable with mean $\mu$.
	For a fixed $k$, we assume $F\left(k\mid \mathbf{X}\right)$ is a monotone decreasing function of $\mu$, which is satisfied for many commonly used models including logistic, Poisson, and negative binomial regression models and ordinal regression with proportionality; see verification in the supplementary material.

	As pointed out in the previous section, $H(s;\mathbf{X})$ is not a continuous function of $\mu$. Consequently,  the density of $H(s;\mathbf{X})$, i.e. $f_{H(s;\mathbf{X})}$, is not smooth. The density of $H( s;\mathbf{X})$ at a point other than $s$, i.e.
	$f_{H( s;\mathbf{X})}(s+\epsilon)$ for $\epsilon\neq 0$, has a different form  from $f_{H( s;\mathbf{X})}(s)$.
	For a small $k$ such as $k\leq 5$ in Figure~\ref{fkxplotapp}, 
	$f_{F\left(k\mid \mathbf{X}\right)}$ contributes to $f_{H( s;\mathbf{X})}$ at $s+\epsilon$  when applying a transformation of random variables.
	While for a large $k$ such as $k=15$,  
	the density of $F\left(15\mid \mathbf{X}\right)$ does not contribute to the density of $f_{H( s;\mathbf{X})}(s+\epsilon)$.  Therefore, in this example,
	\begin{align}\label{eq:densbound}
		\sum_{k=0}^{5}f_{F\left(k\mid \mathbf{X}\right)}(s+\epsilon)\leq f_{H( s;\mathbf{X})}(s+\epsilon)<\sum_{k=0}^{\infty}f_{F\left(k\mid \mathbf{X}\right)}(s+\epsilon).
	\end{align}
	That is, compared with  \eqref{equ:hdensity}, 
	$f_{H( s;\mathbf{X})}$ is not smooth due to loss of $f_{F\left(k\mid \mathbf{X}\right)}$ curves contributing to $f_{H( s;\mathbf{X})}$ at $s+\epsilon$.
	The non-smoothness issue is less of a concern for variables with a finite range. When $\epsilon$ takes a small value $\epsilon_n$ which goes to 0, a finite number of  jump points  would be excluded from the $\epsilon_n$-neighborhood of $s$, 
	 as in Figure~\ref{fkxplotfinite}. 

	To exclude the boundary effect and achieve uniform convergence,  we focus on a closed subset of $(0,1)$ denoted as $[s_L,s_U]$.
	Let 
	$V$ be a subset of $[s_L,s_U]$ such that, for $s\in V$, $g(s)>0$,  which guarantees the availability of data points.  Let the bandwidth  $\epsilon_n$ satisfy that $\epsilon_n\rightarrow0$ and $n\epsilon_n\rightarrow\infty$ as $n\rightarrow\infty$, and  $n\epsilon_n^5=O(1)$.   
	Assume we can interchange the derivatives and the limits, then the  derivatives are 
	$g'(s)=\sum_{k=0}^{\infty}f_{F\left(k\mid \mathbf{X}\right)}'(s)$, $g''(s)=\sum_{k=0}^{\infty}f_{F\left(k\mid \mathbf{X}\right)}''(s)$.  
	Let $K$ be a symmetric kernel function with a compact support,
	and denote $R_2(K)=\int K(u)^2\mathrm{d}u$, $\kappa_2=\int u^2K(u)\mathrm{d}u$. 
	In addition,  we  adopt kernel functions satisfying H\"older conditions with exponent 2, i.e., there is a constant $\alpha_1$ such that 
	$$\mid K(u)-K(v)\mid \leq \alpha_1\mid u-v\mid ^2.$$
	Kernels with a high order of smoothness, e.g., Epanechnikov and quartic  kernels, satisfy this condition.
	
	With the regularity conditions described in Appendix  \ref{sec:assume}, we 
	have the following uniform convergence result of the quasi-empirical residual distribution function, when the model is correctly specified and the estimated parameters are plugged in.

	\begin{theorem}
		\label{pamsee}
		When the regression model is correctly specified, with the estimators of the model parameters plugged in,
		define $\hat{Z}_n(s)=\sqrt{n\epsilon_n}\left(\hat{U}(s;\hat{\bm\beta})-s\right)$.
		Under   Assumptions \ref{monofnew}, \ref{monof} , \ref{lips}, and \ref{bigop},
		and furthermore, if we assume $n\epsilon_n^5\rightarrow R_{\epsilon}$ and adopt kernel functions satisfying H\"older conditions with exponent 2, then
		uniformly in $V$,  $$\hat{Z}_n\rightsquigarrow Z,$$
		where
		$Z$  is a Gaussian process with pointwise mean  $\kappa_2g'(s)\sqrt{R_{\epsilon}}/g(s)$ and variance $R_2(K)s(1-s)/g(s)$. In addition, the covariance of $Z$ is zero  for any two distinct points.
	\end{theorem}

	To analyze the asymptotic results, first, 
	the discrepancy of the proposed function with its null pattern (the identity function) depends on both $m$ and $n$, which is reflected in the pointwise bias $\kappa_2g'(s)\epsilon_n^2/g(s)$ and variance $s(1-s)R_2(K)/\left[n\epsilon_ng(s)\right]$. Even if the  data are highly discrete, i.e. $m$ is small, the bias and variance go to 0 when the sample size is large. In contrast, the deviance residuals are normally distributed  with an error term of order at least $O_p(m^{-1/2})$ which cannot be improved with large sample sizes. Meanwhile, a large $m$ leads to  a large value of $g$ and thereby a small variance.   
	Second, since we use a subset of  data to construct $\hat{U}(\cdot;\hat{ {\bm\beta}})$, the proposed tool has a slower convergence rate than ${n}^{-1/2}$ as for deviance residuals. Therefore, the proposed tool requires a larger sample size for satisfactory performance.

	\section{Simulation}\label{sec:simulation}
	In this section, we use a variety of numerical examples to demonstrate  model assessment using our quasi-empirical residual distribution function.  We examine two important aspects: the proximity  of $\hat{U}(\cdot;\hat{\bm\beta})$ to the null pattern under true models, and its discrepancy with the null pattern under misspecified models. Throughout this section, the bandwidth is selected using the approach proposed in Section \ref{sec:band}, and the Epanechnikov kernel is used.
	\subsection{Closeness to Null Patterns under  Correct Model Specification}\label{sec:truepois}
	
	\textit{Poisson examples.}
	As a valid assessment tool, it is essential to guarantee the closeness to the null pattern if the model is correctly specified. This has been an issue for commonly used residuals including deviance and Pearson residuals when the data are highly discrete. We first explore the effects of the discreteness level, i.e. $m$-asymptotics, through Poisson GLM examples with a log link. Let the location parameter be $\mu =  \beta_0 +X_1\beta_1+X_2\beta_2$, where $X_1 \sim N(0,1)$, and $X_2$ is a dummy variable with probability of 1 as 0.7. The covariates $X_1$ and $X_2$ are independent. We conduct simulations with three levels of discreteness:
	\begin{itemize}
		\item Small mean: $\beta_0 = -2, \beta_1 = 2,\beta_2=1$. 
		\item Medium mean: $\beta_0 = 0, \beta_1 = 2,\beta_2=1$. 
		\item Large mean: $\beta_0 = 5, \beta_1 = 2,\beta_2=1$.
	\end{itemize}
	For each of the experiments, we generate  data, fit the correct regression model, and compute the 
	residuals or quasi-empirical residual distribution function. We then summarize the results graphically by providing the curve of $\hat{U}(\cdot;\hat{ {\bm\beta}})$ and the P-P plots of the residuals. Given a correct model,  the null pattern should be along the diagonal.
	
	Figure~\ref{fig:pois500}  presents the results with sample size 500. 
	The upper row corresponds to the small mean scenario. As anticipated, the deviance and Pearson residuals are far apart from normality due to small $m$, while the proposed quasi-empirical residual distribution function is  close to the identity function. That is, our method provides more reliable conclusions in cases with a high level of discreteness. When we move to the middle row corresponding to the medium mean level, our method keeps the pattern along the diagonal, and the other two residuals get closer to being normally distributed. As the mean increases to the large case (bottom row of Figure~\ref{fig:pois500}), all three methods appear close  to the null pattern. 
	
	\begin{figure}[!h]\centering
	\includegraphics[width=.9\textwidth]{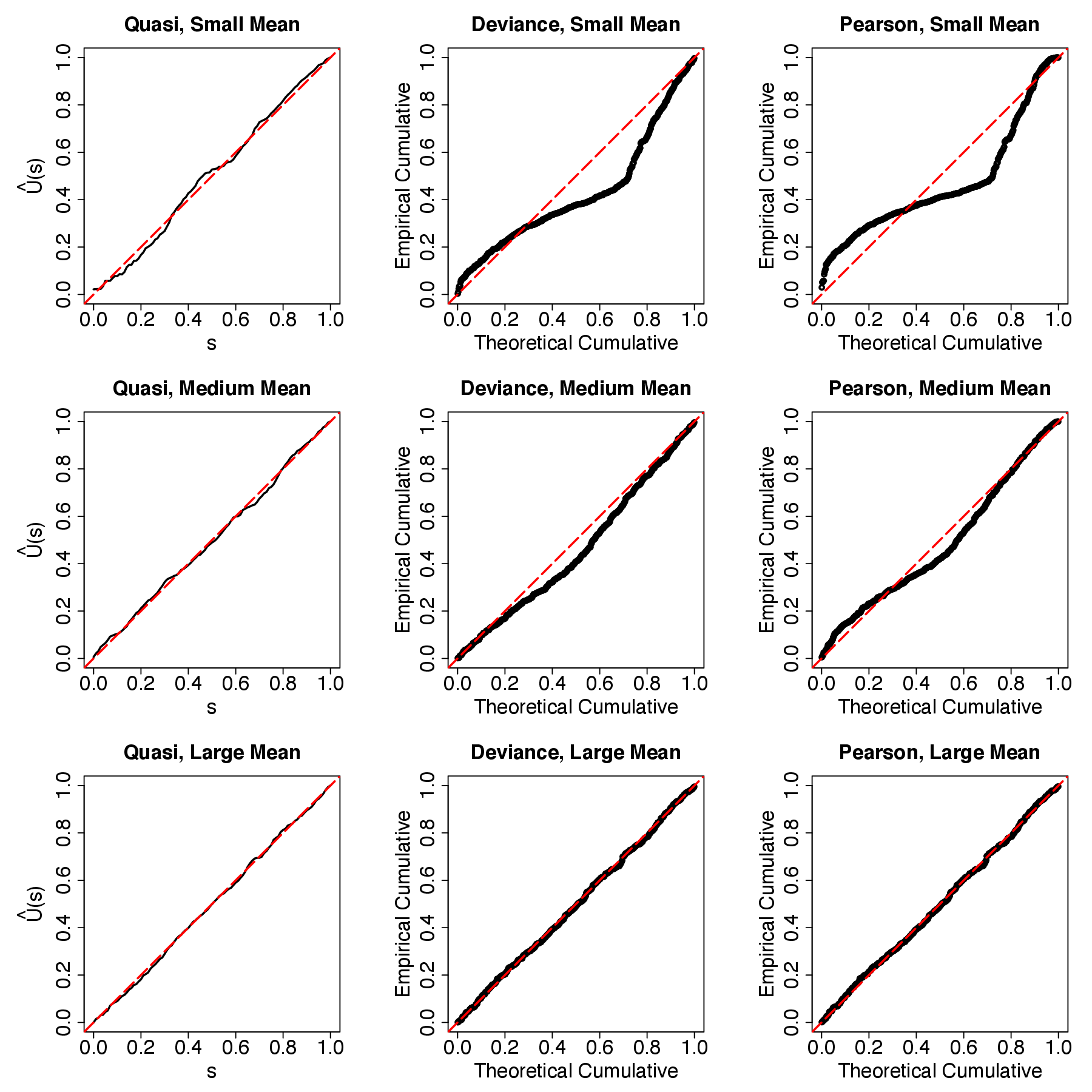}
		
		\caption{Assessment plots for Poisson outcomes under correct models. The three rows correspond to small, medium, and large mean levels. Sample size 500.\label{fig:pois500}}
	\end{figure}

	The demonstrative  results in Figure~\ref{fig:pois500} are based on one randomly selected sample. We further include the results with replications in  Figure~\ref{fig:poisrep500} to visualize the standard error.  The diagonal is within the confidence bands of the proposed method even under a high level of discreteness, though the proposed method has a larger variance under the  small mean scenario compared with the other two types of residuals.
	When the mean increases to medium level, the variance of the proposed tool gets smaller, which is consistent with the theoretical result of Theorem \ref{pamsee}.

		\begin{figure}[!h]\centering\includegraphics[width=.9\textwidth]{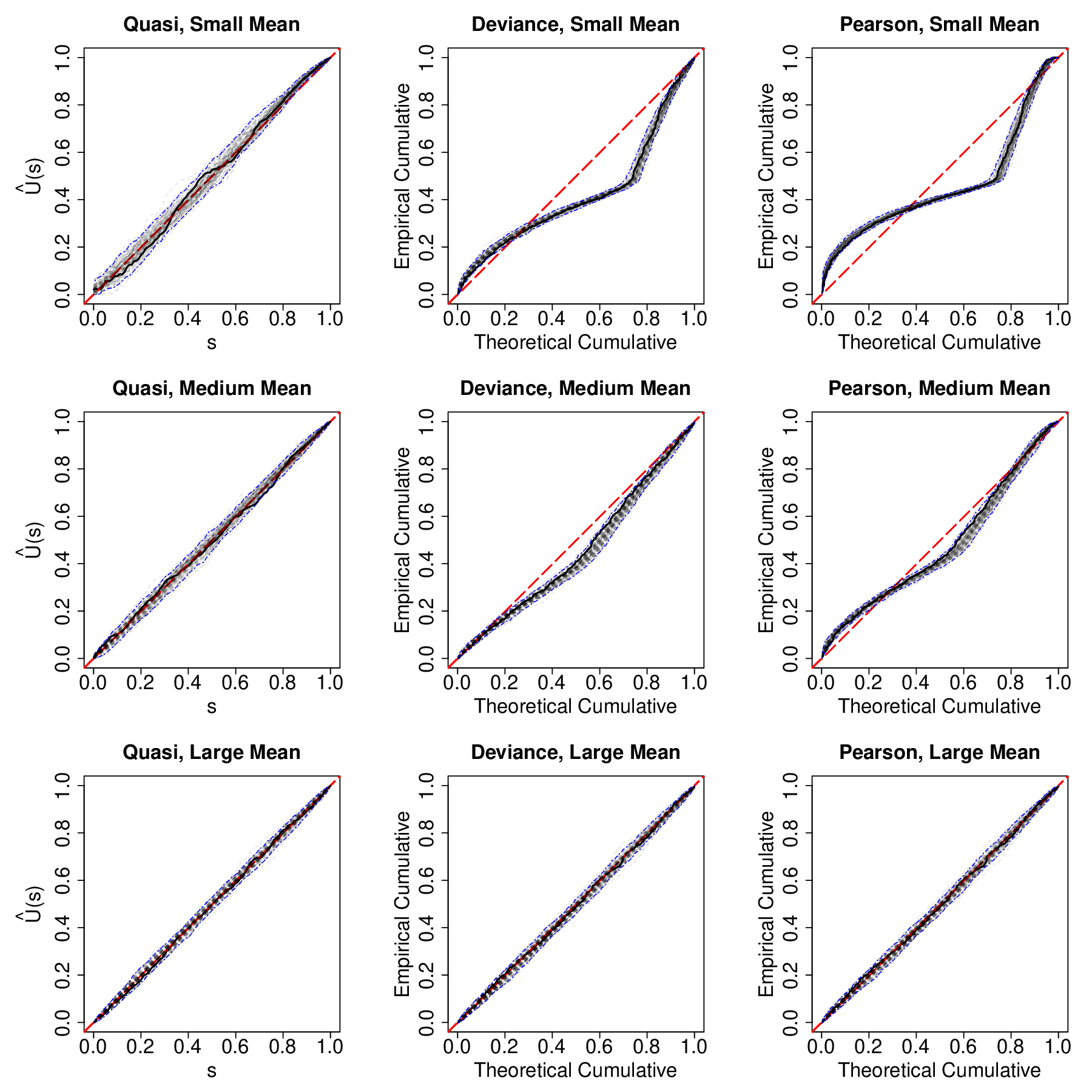}
		
		\caption{Assessment plots for Poisson outcomes under correct models with 100 replications.  Sample size 500.
			The  solid curves correspond to the samples shown in Figure~\ref{fig:pois500},  the  dashed curves are for  replications, and the  dash-dot curves are the 95\% confidence bands.
			\label{fig:poisrep500}}
	\end{figure}

	As discussed in Section \ref{asym}, the proposed tool requires a relatively large sample size for highly discrete data. In the supplementary material, we include  examples of smaller sample sizes, from which we can see that the proposed tool might not be suitable for the cases in which the sample size is smaller than 50, in particular when  data are highly discrete.
	In addition, we  demonstrate in the supplementary material the necessity of at least one continuous covariate for the proposed tool.  Lacking continuous covariates is nevertheless less of an issue if the mean is large.

	\textit{Binary examples.} We now consider an example of binary outcomes with a logistic regression model. 
	We set $\text{logit}(p_1)=\beta_0+\beta_1X_1+\beta_2X_2$ where $p_1$ is the probability of 1, $\beta_0=-2, \beta_1=2,\beta_2=1$, $X_1\sim N(0,1)$, and $X_2$ is a dummy variable with probability of one as 0.7. Figure~\ref{fig:binary} summarizes the results. 
	Due to the undesirable performance of Pearson residuals, we include  randomized quantile residuals instead.
	From the first row,  when the sample size is 100, the difficulty resulting from a high level of discreteness is clear since 
	the deviance  residuals are far apart from normality. However, the proposed method is reasonably close to the null pattern. Randomized quantile residuals do not seem to closely follow a normal distribution as desired in  the top row. This could be the result of a specific realization of the injected noise.
	We further demonstrate this point in Appendix \ref{sec:addsim} (Figure~\ref{fig:random}).
	When we increase the sample size to 2000, 
	the deviance residuals do not improve, whereas the proposed method gets closer to the null pattern. 
	This is consistent with our theoretical results in Section \ref{asym}. For the proposed method, a large sample size is a remedy for the poor performance induced by a high level of discreteness, while the discrepancy of   deviance and Pearson residuals with normality cannot be fixed by increasing the sample size. This property of the proposed method is  especially useful under the current trend of utilizing large datasets.

	\begin{figure}[!h]\centering
		\centering
		\includegraphics[width=.9\textwidth]{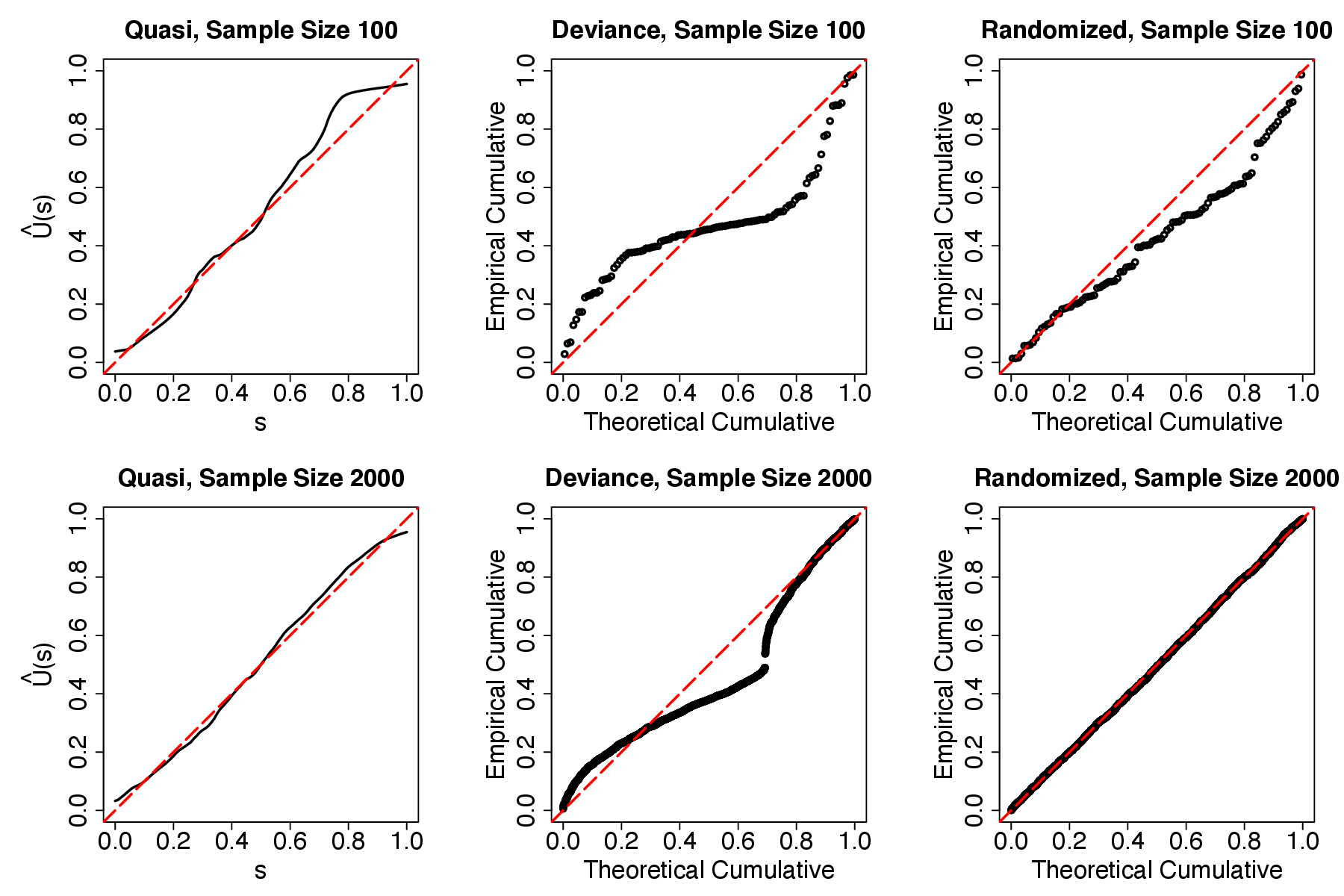} 
		
		\caption{ Assessment plots for binary outcomes under correct models. The two rows correspond to sample sizes 100 and 2000. \label{fig:binary}}
	\end{figure}

	\subsection{Detection of Misspecification}\label{sec:miss}
	As concluded in the previous section, the proposed method is closer to its null pattern under the true model compared with other residuals in various settings. 
	In this section, we check the ability of the proposed method to detect common causes of model misspecification including omission of covariates,  overdispersion, and incorrect link functions for count data. Experiments on diagnosing underdispersion for count data and 
	 non-proportionality in ordinal regression are included in the supplementary material.
	
	\textit{Missing covariates.} We first demonstrate that our quasi-empirical residual distribution function is a useful  tool for detecting missing covariates through a Poisson example. The underlying location parameter is  $\mu =  \beta_0 +X_1\beta_1+X_2\beta_2$, where $X_1,X_2 \sim N(0,1)$ independently, and $\beta_0=-2, \beta_1=2,\beta_2=1.5$. 
	Under the misspecified model, the covariate $X_2$ is missing. Figure~\ref{fig:omitsmall500} includes the results 
	in which we compare the proposed method with deviance  and randomized quantile residuals. 
	By comparing the top row which is for the true model with the bottom row corresponding to the misspecified model, we can see that our tool is illuminating in the sense that the curve is close to the null pattern under the true model, while it shows a large discrepancy when a covariate is missing, even with a high level of discreteness. For deviance residuals,   there is no significant improvement from the misspecified model to the true model, and thus it is hard to draw a conclusion about whether the model is sufficient for the data. Randomized quantile residuals show a slight discrepancy under the misspecified model.
	The results for medium mean level with $\beta_0=0$ is included in the the supplementary material,
	wherein all the methods become more informative, and the proposed method still outperforms.
	We also include an example of the proposed tool detecting the missingness of a  quadratic term in the supplementary material.
	
	\begin{figure}[!h]\centering
		\centering\includegraphics[width=.9\textwidth]{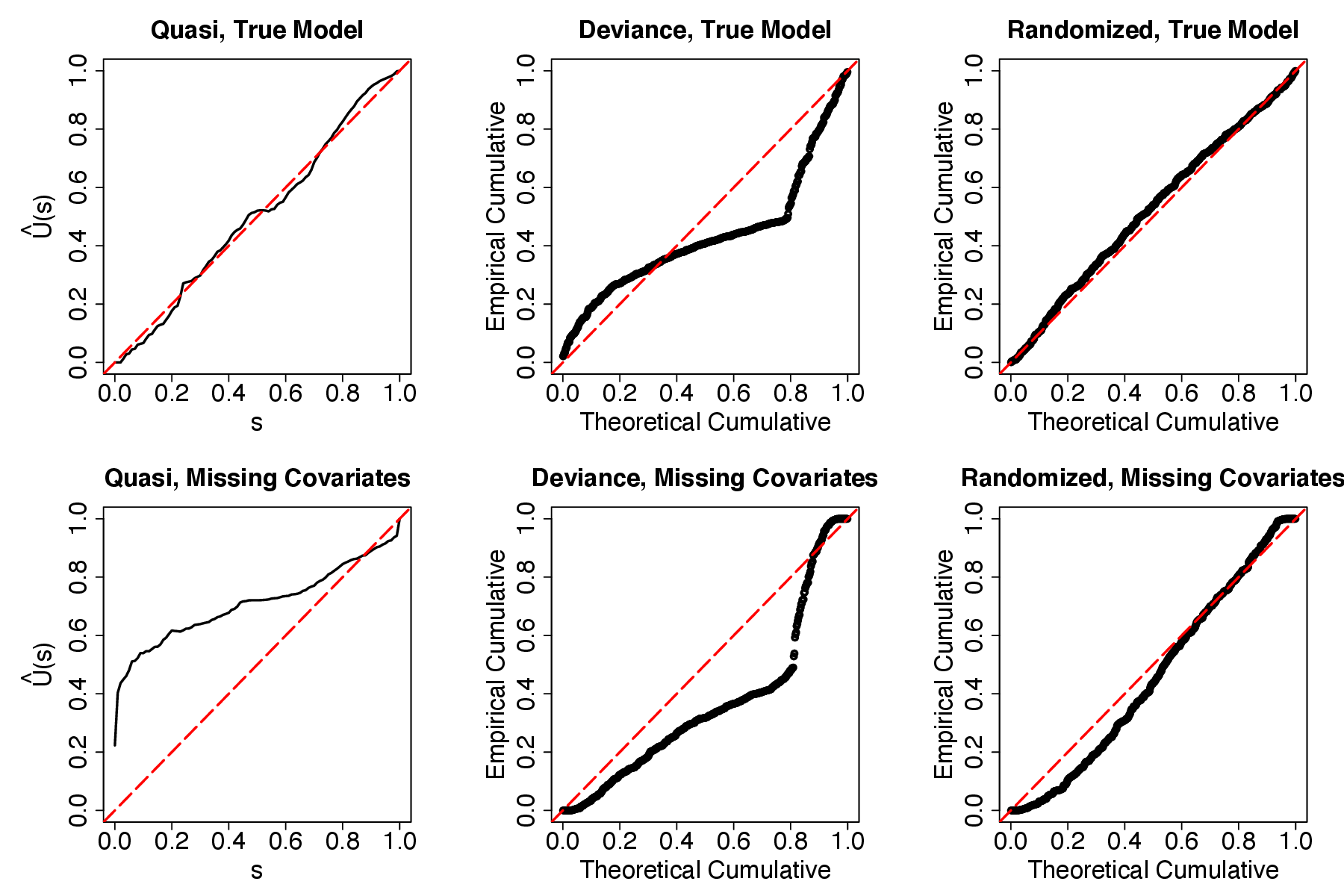}
		\caption{Assessment  plots for Poisson outcomes when a covariate is missing in the small mean scenario. Top row: correct model. Bottom row: a covariate is omitted. The sample size is 500.\label{fig:omitsmall500}}
	\end{figure}

	\textit{Overdispersion.} For discrete outcomes, especially count data, one of the most common issues  is overdispersion.  We next use a numerical example to examine the ability of the proposed method as well as other residuals to detect overdispersion. 
	We  generate data using a negative binomial distribution with the same mean structure as  the Poisson outcomes in Section \ref{sec:truepois}, and the size parameter is set to be 2. 
	For the misspecified model, we fit the data with a Poisson GLM and thus overdispersion is present.
	Figure~\ref{fig:dispsmall500} shows the results in the small mean scenario. The top row includes the results for the true model, while the bottom row shows the results under the misspecified model. By examining the first column, 
	we can see that the proposed tool is again informative. 
	In contrast, the deviance residuals show a large discrepancy under both models, while randomized quantile residuals are not sensitive to misspecification in this case. 
	Besides overdispersion, our method can also detect the issue of underdispersion, which we illustrate in the supplementary material.
	
	\begin{figure}[!h]\centering
		\centering
\includegraphics[width=.9\textwidth]{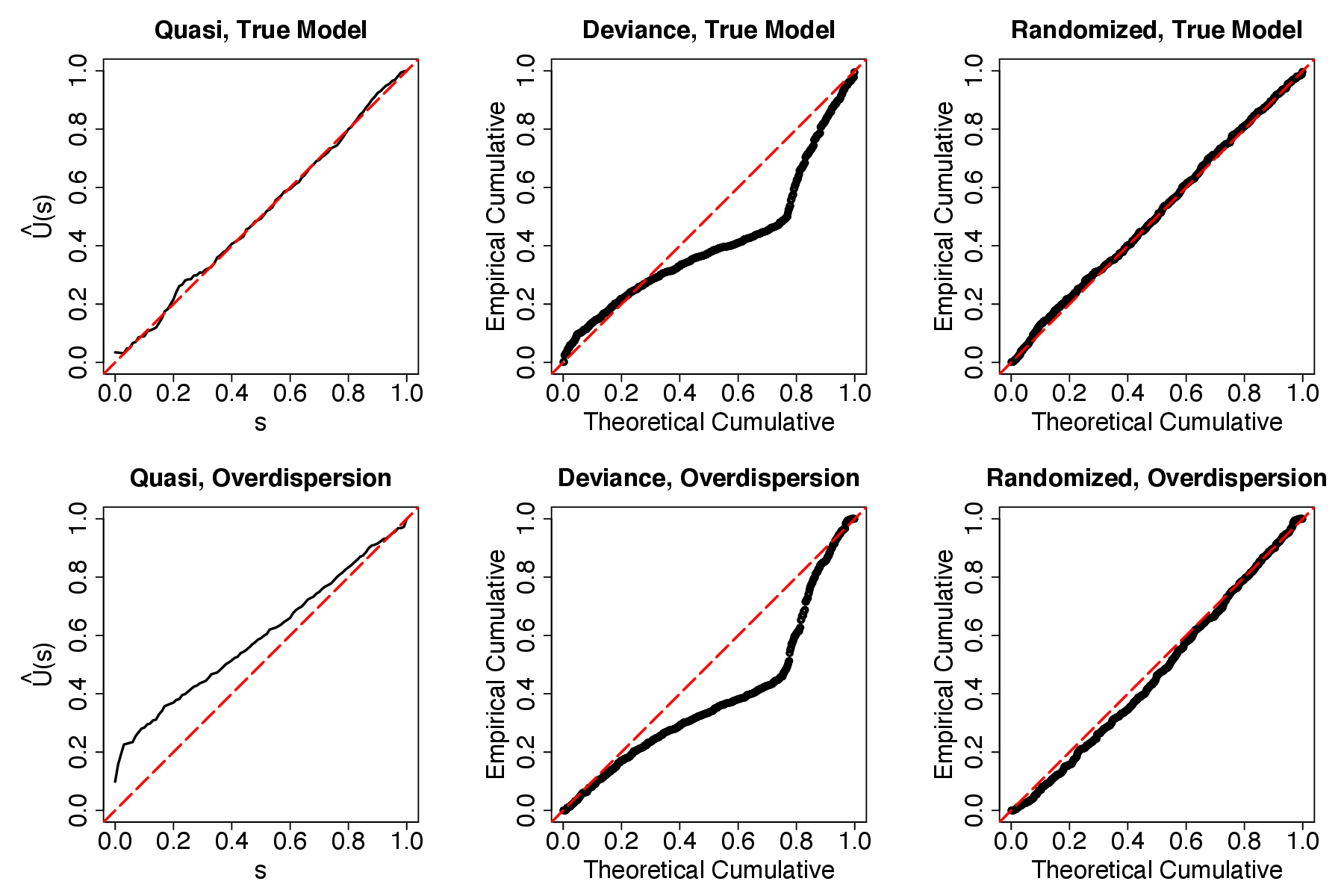}
		
		\caption{Graphical detection for overdispersion in negative binomial outcomes under the small mean scenario. The top row corresponds to the true model, and the bottom row shows the results when a Poisson GLM is mistakenly  used. Sample size: 500. \label{fig:dispsmall500}}
	\end{figure}

\textit{Zero inflation.}
Zero-inflated Poisson models are also commonly adopted to tackle  overdispersed data.  We now include an example of  zero-inflated Poisson  models to demonstrate the  usage of the proposed method for non-exponential distributions. 
The probability of excess zero is modeled with $\mathrm{logit}(p_0)=\beta_{00}+\beta_{10}X_1$, and the Poisson component has a mean $\lambda=\exp\left(\beta_0+\beta_1X_1+\beta_2X_2\right)$, where $X_1\sim N(0,1)$ and $X_2$ is a dummy variable with probability of 1 as 0.7, and $(\beta_{00},\beta_{10},\beta_0,\beta_1,\beta_2)=(-2,2,0,2,1)$. We compare  the true model with  a Poisson model using the proposed tool. 
Figure~\ref{fig:zeroinf500med} shows the results from which we can see the proposed method and randomized quantile residuals can help identify the insufficiency of fitting. In contrast, Pearson residuals show a large discrepancy over the whole range under the true model and thereby is not revealing.

\begin{figure}[!h]\centering
	\centering
		\includegraphics[width=.9\textwidth]{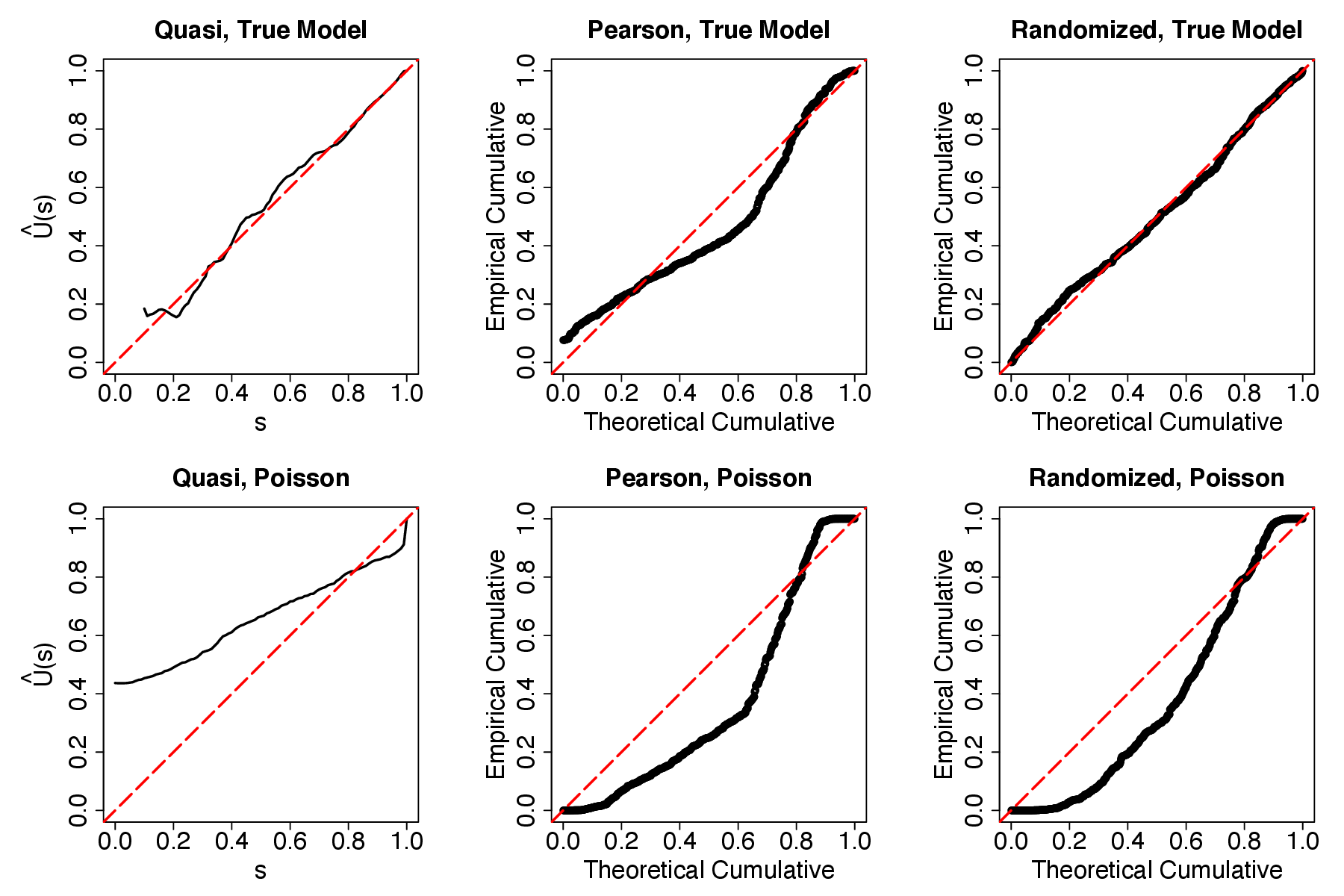}
	
	\caption{Graphical assessment of models with zero-inflated Poisson outcomes. The top row corresponds to the true model, and the bottom row shows the results when the data are fit with a Poisson GLM. Sample size 500. \label{fig:zeroinf500med}}
\end{figure}

It is noticeable that under the correctly specified model, the curve of $\hat{U}(\cdot;\hat{ {\bm\beta}})$ seems  unstable at the lower left corner.  
Figure~\ref{fig:zeroinfhistp0} includes the histogram of $F(0\mid\mathbf{X})=p_0+(1-p_0)\exp(-\lambda)$, from which we can see that there is no observation below 0.16. This explains the aberrant behavior of $\hat{U}(s;\hat{ {\bm\beta}})$ for small $s$ values, which should be downweighted.
We also notice that  $\hat{U}(\cdot;\hat{\bm{\beta}})$ is not monotone, in particular in the area with sparse data.
Asymptotically, $\hat{U}(\cdot;\hat{\bm{\beta}})$ converges to a monotonic function, the identity function. However, with finite samples,
it is not guaranteed that $\hat{U}(\cdot;\hat{\bm{\beta}})$ is monotone. The techniques for nonparametric regression with monotonicity constraints (e.g. \citealt{hall2001nonparametric}) might be a remedy for this matter with finite samples; we leave it as future work.

\begin{figure}[!h]\centering

\includegraphics[width=.3\textwidth]{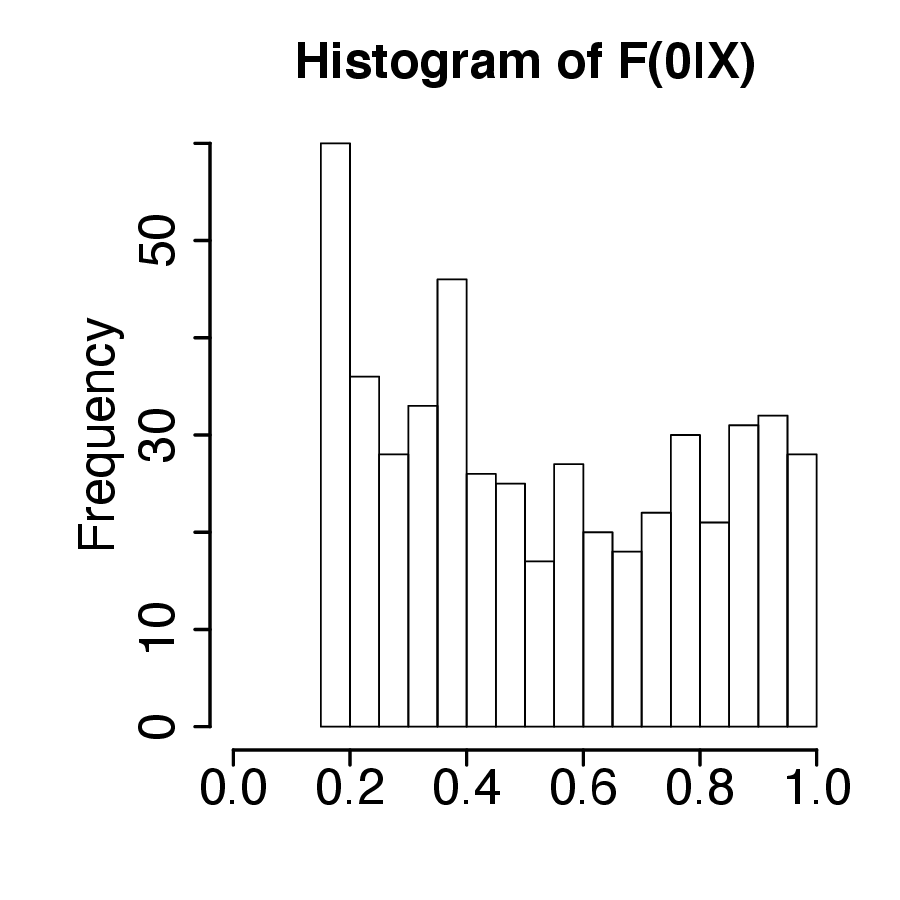}
	
	\caption{Histogram of the probabilities of  zero  in the zero-inflated Poisson example in Figure~\ref{fig:zeroinf500med}. \label{fig:zeroinfhistp0}}
\end{figure}

	\textit{Incorrect link function.} To provide an assessment of link functions, we  present a Poisson example. The true link function is the square root function, i.e., the mean $\lambda=\left(\beta_0+\beta_1X_1+\beta_2X_2\right)^2$ where $X_1\sim N(0,1)$ and $X_2$ is a binary variable with probability of 1 as 0.7, and $(\beta_0,\beta_1, \beta_2)=(0,1,1)$. In the misspecified model, the log link is used. The comparative results are included in Figure~\ref{fig:link500med}. By comparing across the two rows, we can see the proposed method and the randomized quantile residuals show the transition from being close to the diagonal to a disagreement. The deviance residuals also show a larger discrepancy with an incorrect link function, though it has a noticeable difference with the null pattern under the true model.  Nonetheless we also notice that in many scenarios, the log link can provide reasonable fitting even if the true link function is the square root or identity function.

	\begin{figure}[!h]\centering\includegraphics[width=.9\textwidth]{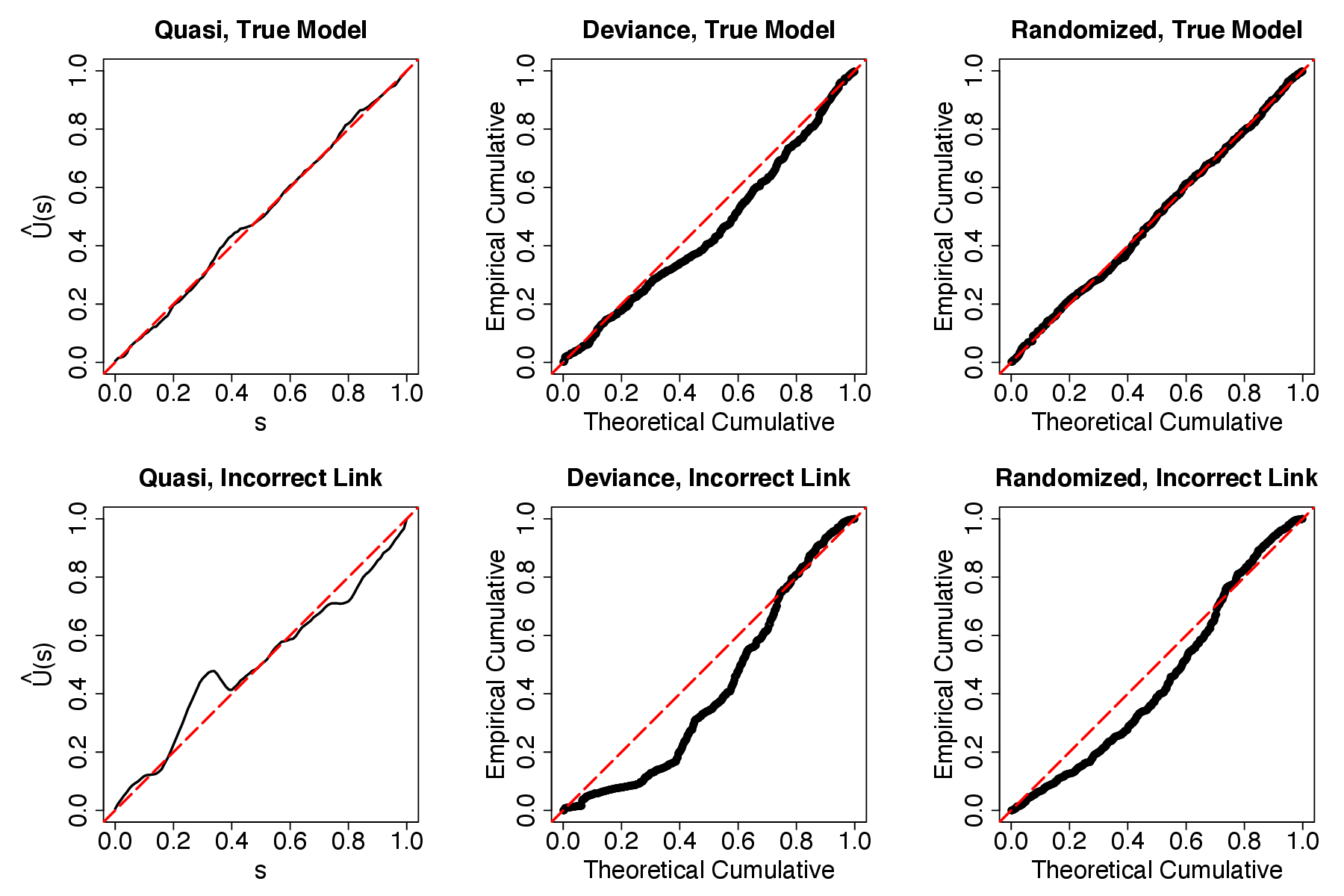}
		
		\caption{Graphical assessment  of link function for Poisson models. Sample size 500. The two rows correspond to right link (square root) and wrong link (log). \label{fig:link500med}}
	\end{figure}
	\subsection{Possible Explanations for Discrepancies}\label{sec:reason}
	In the previous sections, we showed that the proposed method can provide an accurate  signal on model adequacy, superior to the established residuals. In this section, we provide several examples to unravel the discrepancies between $\hat{U}(\cdot;\hat{\bm{\beta}})$ and the identity function.
	
	From the construction of $\hat{U}(\cdot;\hat{\bm{\beta}})$, for a given  value of $s$ and  observation $i$, under the fitted distribution function denoted by $\hat{F}$, 
	the observation is assigned with a nonzero weight if
	there exists an integer $k$ such that $\hat{F}(k\mid\mathbf{X}_i)\approx s$.
	The resulting $\hat{U}(s;\hat{\bm{\beta}})$ is an estimator of the probability $F(k\mid\mathbf{X}_i)$ for the selected observations.
	If the curve is above the identity, it indicates  ${F}(k\mid\mathbf{X}_i)>s\approx \hat{F}(k\mid\mathbf{X}_i)$. 
	On the other hand, if the curve is below the identity function, it indicates that ${F}(k\mid\mathbf{X}_i)<\hat{F}(k\mid\mathbf{X}_i)$. Therefore, the shape of the resulting $\hat{U}(\cdot;\hat{\bm{\beta}})$ curve reflects  the relative relationship between the underlying and fitted distribution functions.
	
	\textit{Intercept.}
	For unambiguous  illustration, we investigate an example wherein  only the intercept is incorrect and other parameters are fixed. The true model is a Poisson GLM with mean $\lambda=\exp\left(\beta_0+\beta_1X_1+\beta_2 X_2\right)$, where $X_1\sim N(0,1)$,  $X_2$ is a binary variable with probability of 1 as 0.7, and $(\beta_0,\beta_1,\beta_2)=(0,2,1)$. When the model is correctly specified, we see a curve of $\hat{U}(\cdot;\hat{\bm{\beta}})$ along the diagonal  in the left panel of Figure~\ref{fig:poisoverunder}. When the intercept is mistakenly set to be 2, the mean is overestimated, and as a result, for any given $k$, ${F}(k)>\hat{F}(k)$; see discussion in Section \ref{asym}. 
	One can see that in the middle panel, the  resultant curve is above the diagonal. On the other hand, when the intercept is underestimated to be $-2$, the curve is below the diagonal.  The coefficient for a quadratic term has the similar effect, which we illustrate in the supplementary material. 
	
	\begin{figure}[!h]\centering\centering\includegraphics[width=.9\textwidth]{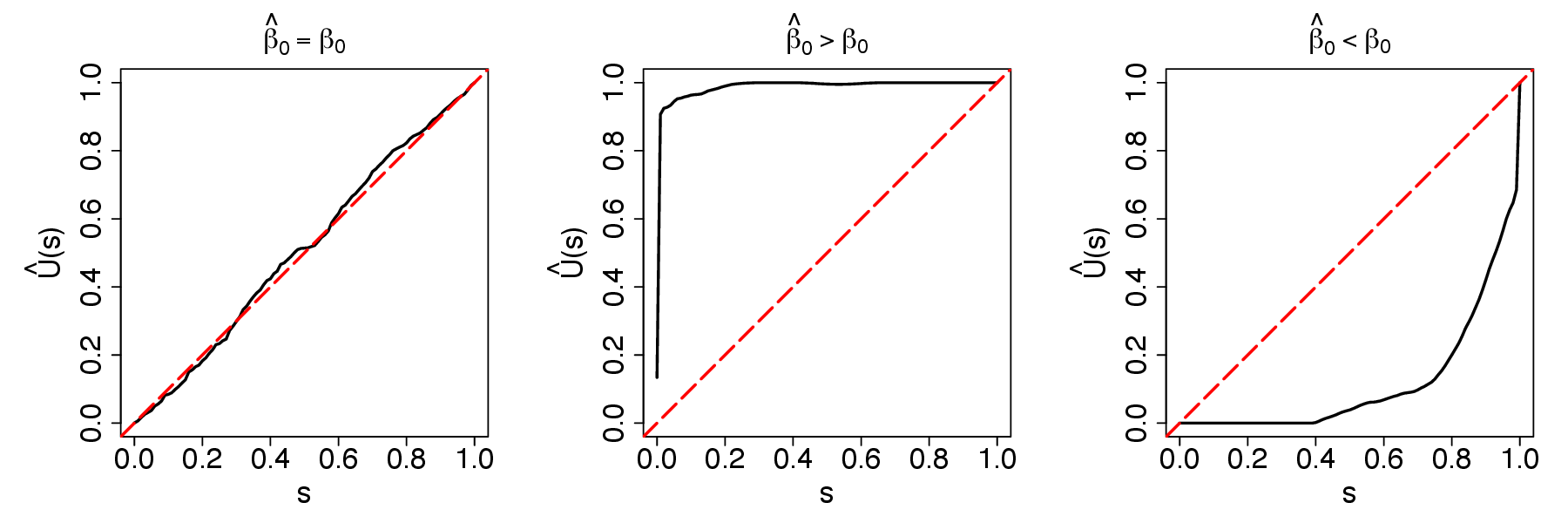}
		\caption{Graphical assessment of intercept in a  Poisson GLM. The sample size is 500.\label{fig:poisoverunder}}
	\end{figure}

	
	\textit{Overdispersion.} In GLM frameworks, the mean structure is often close to being correctly estimated even if 
	overdispersion is an issue (\citealt{mccullagh1989generalized}). In  the supplementary material, we explore the distribution function 
	of Poisson and negative binomial variables, when their   means are the same. If the  mean is  bigger than 1,  the  distribution function of the negative binomial distribution is mostly bigger than the Poisson distribution, which explains the shape of $\hat{U}(\cdot;\hat{\bm{\beta}})$ in Figure~\ref{fig:dispsmall500}.

	Our simulation studies demonstrate that our tool  provides an accurate signal on  model adequacy in various settings. In addition, the magnitude of the discrepancy between our curve and the diagonal carries information about the severity of misspecification.
	However, given its  nature  as a surrogate for the empirical residual distribution function rather than residuals themselves, our tool is not able to  identify the exact causes of misspecification. 
	Should the model be insufficient according to our method, 
	tools devoted to specific diagnostic tasks can be further employed. 
	Regression model-building process should involve iterations between assessment plots and refinements of the model, until reasonable agreement between the model and the data is reached
	(\citealt{cook1999graphs}).
	This final step can be confirmed by our assessment tool.

	\section{Analysis of Insurance Claim Frequency Data}\label{sec:data}
	In this section, we present an application of the proposed quasi-empirical residual distribution function to insurance claim frequency data. Frequency,  the number of reported claims from each policyholder, is an important component of insurance claim data and largely reveals the riskiness of a policyholder. Here, we use a dataset from the Local Government Property Insurance Fund (LGPIF) in the state of Wisconsin, USA.
	The LGPIF was established by Wisconsin government to provide property insurance for local government entities. 
	In this paper, we focus on building and contents (BC) insurance, which is the major coverage offered by the LGPIF. The dataset  contains 5660 observations from year 2006 to 2010. Table~\ref{table:empiricalcount} provides the empirical numbers of observations.  Covariates together with their summary statistics are displayed in Appendix \ref{sec:addsim} Table~\ref{tab:covariates}. Among them, the amounts of insurance coverage and deductible are continuous covariates, which are necessary for applying the proposed tool.

	\begin{table}[!h]\centering 
		\caption{Distribution  of the number of claims in the LGPIF data.
			\label{table:empiricalcount} }
		\begin{tabular}{@{\extracolsep{5pt}} cccccccc} 
			\\[-1.8ex]\hline 
			\hline \\[-1.8ex] 
			Total&0&1&2&3&4&5&$>$5\\
			\midrule
			$5660$ & $3976$ & $997$ & $333$ & $136$ & $76$ & $31$ & $111$ \\
			\bottomrule
		\end{tabular} 
	\end{table}

	We fit several commonly used count regression models to the claim frequency data: Poisson, negative binomial (NB),  zero-inflated Poisson, and zero-inflated negative binomial. 
	In addition, it can be seen from Table~\ref{table:empiricalcount} that the data contain a large number of zeros and a significant amount of ones. This motivates the usage of a zero-one-inflated Poisson model as described in \cite{frees2016multivariate}. Its distribution  function  can be expressed as 
	$$F(k)=
	\begin{cases*}
	\pi_{0}+(1-\pi_{0}-\pi_{1})\exp(-\lambda)&$k=0$,\\
	\pi_{0}+\pi_{1}+(1-\pi_{0}-\pi_{1})\sum_{i=0}^k \lambda^i\exp(-\lambda)\frac{1}{i!}&$k>0$,
	\end{cases*}$$
	where $\pi_0$ and $\pi_1$ are the probabilities of extra zeros and ones, respectively, and $\lambda$ is the  Poisson mean.
	
	We then apply the proposed assessment tool to the models. Figure~\ref{fig:data} displays   the quasi-empirical residual distribution function (solid curves). 
	The plots suggest that the zero-one-inflated Poisson  and the negative binomial models provide  satisfactory fitting, while the Poisson and  zero-inflated Poisson models fit the data poorly. 
	Interestingly, the zero-inflated negative binomial model is not as good as the negative binomial model. We notice that the zero-inflated negative binomial model focuses on zeros and tends to underestimate the probability of ones, which explains the deficiency of its fitting.
In addition, we present the confidence bands of the quasi-empirical residual distribution function constructed using bootstrap with 500  replications (dash-dot curves). We observe that the negative binomial model has wide and wiggly confidence bands in the lower left corner, indicating this model might  be unstable.
	
	\begin{figure}[!h]\centering	\includegraphics[width=.9\textwidth]{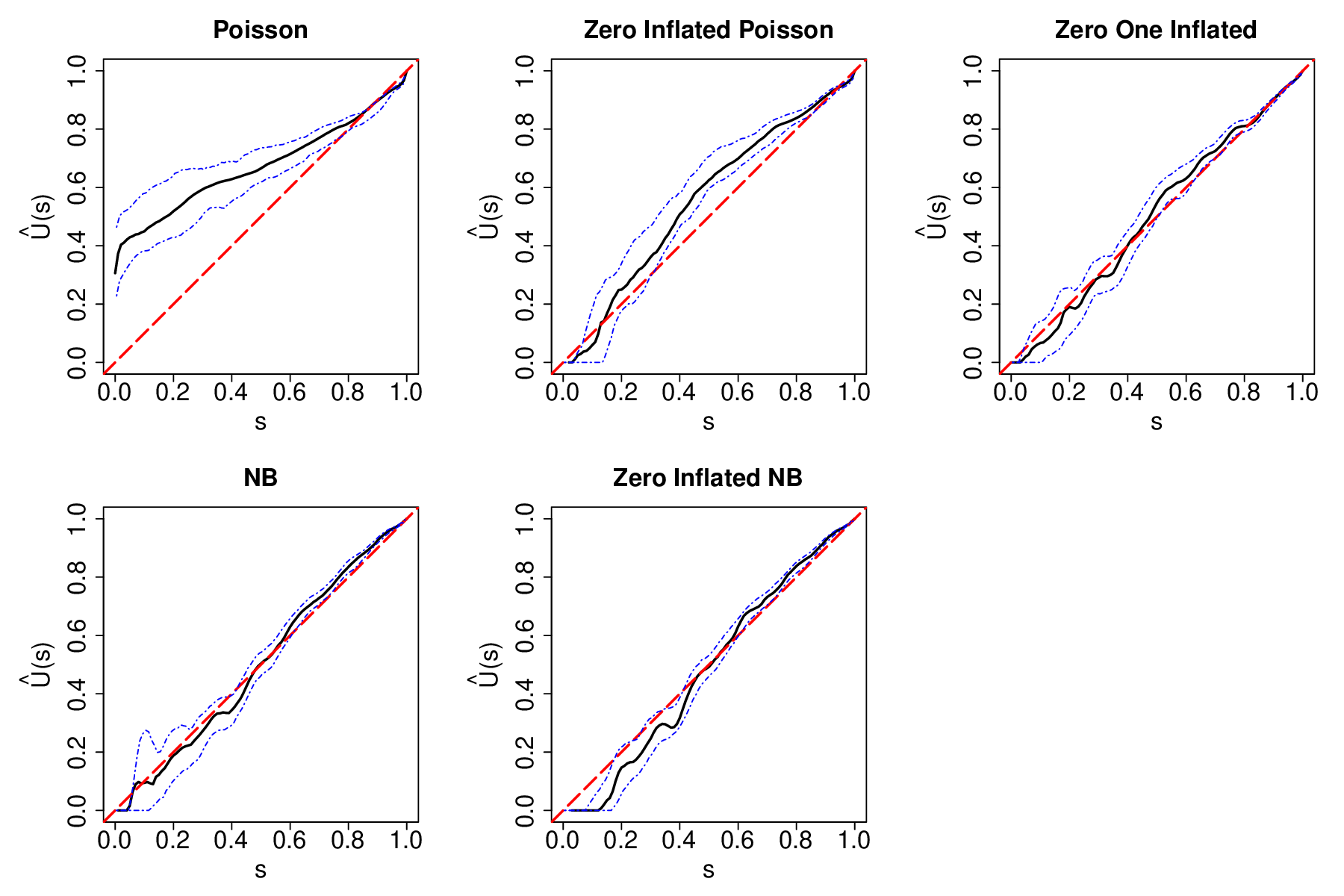}
		
		\caption{Graphical assessments  of Poisson, zero-inflated Poisson,  zero-one-inflated Poisson, negative binomial, and zero-inflated negative binomial models on the LGPIF data. \label{fig:data}}
	\end{figure}


To summarize the models' goodness-of-fit  numerically, in
Table~\ref{tab:dist}, we provide  the $L_2$-norm distances between $\hat{U}(\cdot;\hat{\bm\beta})$ and the diagonal
\begin{align}\label{eq:l2}	\left\lbrace\int_{[0.1,1]}\left(\hat{U}(s,\hat{ {\bm\beta}})-s\right)^2\mathrm{d}s\right\rbrace ^{1/2}.
\end{align}
Here, we integrate over $[0.1,1]$ since there is no grid value below $0.1$ for the zero-inflated Poisson model.  A good model should have a small $L_2$-norm distance.
We can see that the 
negative binomial model outperforms other approaches with smallest distance, followed closely by the zero-one-inflated Poisson model. 
Table~\ref{tab:dist} also contains the standard deviations of the distances across the 500 bootstrap replications. Consistent with our visual impression, 
the negative binomial model has
 a bigger standard deviation compared with the zero-one-inflated Poisson model.

\begin{table}[!htbp] \centering 
	\caption{ $L_2$-norm distances between the quasi-empirical distribution and the identity function \eqref{eq:l2}  (with their bootstrap standard errors in parenthesis) and chi-square goodness-of-fit statistics of different  models for the LGPIF data.
				\label{tab:dist} }
	\begin{tabular}{ cccccc} 
		\\[-1.8ex]\hline 
		\hline \\[-1.8ex] 
		
	&Poisson&0-Inflated Poisson&	0-1-Inflated Poisson&NB&0-Inflated  NB\\
			\midrule
	$L_2~(\times 100)$	&$18.615$ & $7.200$ & $2.775$ & $2.644$ & $4.958$ \\ 
	&(2.334) &(1.416)& (0.521) &(0.660) &(0.741)\\
Chi-square	&105.201& 154.572&  77.053 & 88.086&  98.398\\
	\bottomrule
	\end{tabular} 
\end{table} 
	
	To summarize, the negative binomial model provides best fitting for our data in terms of the $L_2$-norm distance, and it is a  more parsimonious  model  than the zero-one-inflated Poisson. On the other hand, the negative binomial model  seems less stable, reflected by its wider confidence band.
	We thereby recommend to pursue both negative binomial and zero-one-inflated Poisson models, whose 
	coefficients  are provided in Table~\ref{table:coef}.
	Our conclusion is mostly consistent with the model selection results  in \cite{frees2016multivariate} by means of chi-square goodness-of-fit  statistics (the second row of Table~\ref{tab:dist}). However, whether the best models among candidates are sufficient  remained unclear  in \cite{frees2016multivariate}. Using our tool, we can make  informative conclusions on the adequacy of the fitted models based on Figure~\ref{fig:data}, which indicates that the selected models indeed fit the data reasonably well.

	\section{Conclusions}\label{sec:conc}
	In this paper, we proposed a quasi-empirical residual distribution function for assessing regression models  with discrete outcomes. We showed the uniform convergence of the proposed tool under the correctly specified model. Through simulation studies and empirical analysis, we demonstrated that the proposed method has appealing properties,  when at least one continuous covariate is available. In particular, we showed that the quasi-empirical residual distribution function is close to the hypothesized pattern under the true model, and under  misspecified models, it shows  significant discrepancies. It was also highlighted that, even under  a high level of discreteness (e.g., binary outcomes and Poisson outcomes with small means), the proposed method gives reasonable results. In addition, as the sample size increases, its performance  improves; whereas for  commonly used assessment tools such as deviance residuals, there is a significant error term which cannot be fixed by a large sample size.

	The proposed tool has some limitations. First, it is not applicable without at least one continuous covariate, except for large mean scenarios.  Second, our tool cannot be used to generate residual versus predictor plots, and hence it may  not be possible to uniquely identify certain causes of misspecification such as missing covariates and outliers.
	One can combine our approach with  tools devoted to specific diagnostic tasks in practice. 
	Third, since our tool is based on  a subset of  data, it has a slower convergence rate than ${n}^{-1/2}$ and thus might require a relatively large sample size for satisfactory performance, in particular if data are highly discrete.
	Lastly, the proposed tool   is more complicated than other established assessment tools, as it requires tuning of the bandwidth. 

	Besides   graphical assessments of regression models, the proposed tool may be a feasible starting point for the  construction of goodness-of-fit tests to obtain conclusions with statements of statistical confidence.
	This is known to be a quite challenging issue for discrete outcomes (\citealt{mccullagh1986conditional}).
	The weak convergence results of this paper build the essential foundation for goodness-of-fit tests, however we leave this as a direction of future research.

\begin{appendices}

		\section{Additional Simulation and  Data Analysis   Results}\label{sec:addsim}

Randomized quantile residuals are based on the idea of continuization. 
However, in order to produce continuous variables, 
randomness is introduced by adding an external uniform random variable, as discussed in Section \ref{sec:intro}.
In Figure~\ref{fig:random}, 
for a given dataset and a given model, the randomized quantile residuals give contradictory conclusions with different random seeds (middle and right panels). The effects of randomness would vanish as the sample size increases or the discreteness level reduces.

\begin{figure}[!h]\centering
	\includegraphics[width=.9\textwidth]{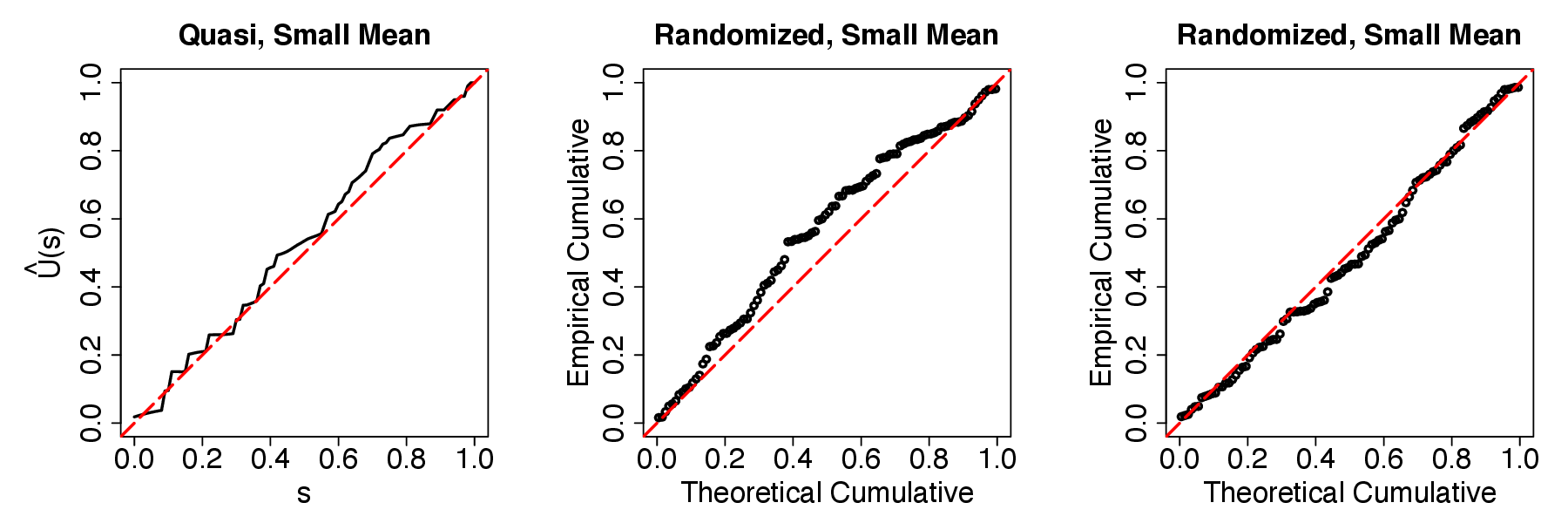} 
	
	\caption{Assessment plots for Poisson outcomes with a high discreteness level using the proposed method (left) and randomized quantile residuals with two different random seeds (middle and right). Sample size 100.\label{fig:random}}
\end{figure}


Figure~\ref{fig:poisband} demonstrates the sensitivity of the quasi-empirical residual distribution function  to different bandwidths. We can see that the bandwidth plays more of an important role under the small mean scenario. If the bandwidth is too small, $\hat{U}(\cdot;\hat{\bm\beta})$ could end up looking quite wiggly (dotted curves).  If the bandwidth is too big, on the other hand, it might induces bias between $\hat{U}(\cdot;\hat{\bm\beta})$ (dash-dot  curves) and its true pattern. The  dashed lines correspond to the results of the proposed bandwidth selector (Section \ref{sec:band}) which appear to be reasonably good. 
\begin{figure}[!h]\centering
	\includegraphics[width=.9\textwidth]{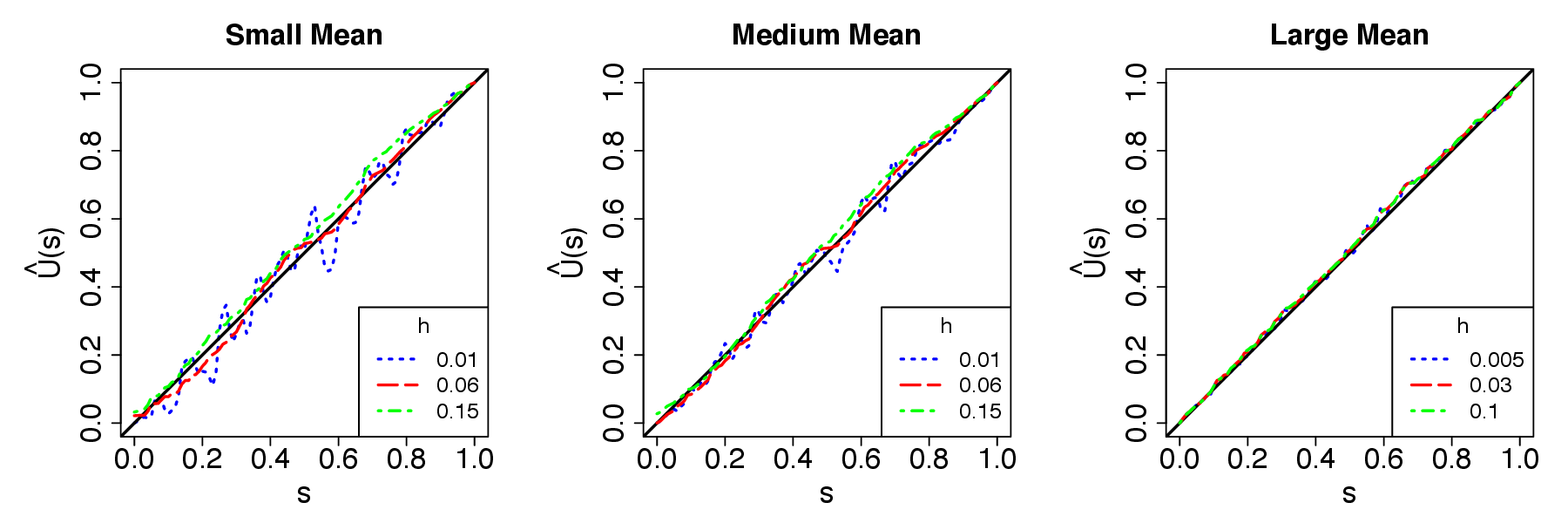} 
	\caption{Sensitivity  of $\hat{U}(\cdot;\bm{\beta})$ to the bandwidth, with undersmoothing (dotted),  oversmoothing (dash-dot) and selected bandwidths (dashed). Sample size 500.\label{fig:poisband}}
\end{figure}

Table~\ref{tab:covariates} summarizes the covariates in the LGPIF data.  Table~\ref{table:coef} contains the fitted coefficients of the selected models.
\begin{table}
	\centering
	\caption{Description and summary statistics of covariates in the LGPIF data.\label{tab:covariates}}
	\begin{tabular}{llc}
		\\[-1.8ex]\hline 
		\hline \\[-1.8ex] 
		Variable      & Description      & Mean (s.d.)\\
		\midrule
		TypeCity &   =1 if entity type is city    & 0.140 \\
		TypeCounty &  =1 if entity type is county     & 0.058  \\
		TypeSchool &  =1 if entity type is school     & 0.283 \\
		TypeTown &  =1 if entity type is town     & 0.173  \\
		TypeVillage &   =1 if entity type is village    & 0.237  \\
		TypeMisc &   =1 if entity type is other   &0.108\\
		NoClaimCredit&   =1 if  no building and content claims\\& in prior year  &0.329\\
		lnCoverage & Coverage of BC line in logarithmic &   2.119\\&millions of dollars  & (2.000)      \\
		lnDeduct&BC deductible level  in logarithmic & 7.155\\&millions of dollars  &(1.174)\\
		\bottomrule
	\end{tabular}%
	
\end{table}%

\begin{table}
	\centering
	
	\caption{Estimated coefficients and dispersion parameter of the selected regression models (negative binomial and zero-one-inflated Poisson) for the LGPIF data.\label{table:coef}}
	\begin{tabular}{@{\extracolsep{5pt}}llrrrr}
		\toprule
		&&\multicolumn{2}{c}{0-1 inflated Poisson}&\multicolumn{2}{c}{Negative Binomial}\\
		\midrule
		&&Coef.&s.e.&Coef.&s.e.\\
		\midrule
		Count&(Intercept) & -1.540 & 0.125 & -0.945 & $0.161$ \\
		
		&lnCoverage & 0.751 & 0.023 & $0.923$ & $0.030$ \\
		
		&lnDeduct & -0.020 & 0.017& -0.247 & $0.025$   \\
		
		&NoClaimCredit & -0.395 & 0.131 & -0.478 & $0.069$ \\
		
		&TypeCity & -0.143 & 0.079  & -0.198 & $0.090$ \\
		
		&TypeCounty & -0.250 & 0.087 & -0.203 & $0.113$  \\
		
		&TypeMisc & -0.195 & 0.179  & -0.598 & $0.139$\\
		
		&TypeSchool & -1.157 & 0.085 & -0.985 & $0.088$  \\
		
		&TypeTown & 0.186 & 0.175&$0.144$ & $0.137$   \\
		\midrule
		Zero&(Intercept) & -4.755 & 0.448 \\
		
		&lnCoverage & -0.580 & 0.078  \\
		
		&lnDeduct & 0.879 & 0.062   \\
		
		&NoClaimCredit & 0.536 & 0.280  \\
		
		\midrule
		One&(Intercept) & -5.533 & 0.639   \\
		
		&lnCoverage & -0.047 & 0.094   \\
		
		&lnDeduct & 0.577 & 0.084   \\
		
		&NoClaimCredit & 0.300 & 0.353   \\
		\midrule
		Dispersion&&&&0.703\\
		\bottomrule
	\end{tabular}%

\end{table}%
 \section{
	Additional Theoretical Results}\label{sec:assume}
		To formalize the discontinuity pattern of $H( s;\mathbf{X})$,
		denote $M_s^k$ as the jump point of $H( s;\mathbf{X})$  transiting from $F(k-1\mid \mathbf{X})$ to $F\left(k\mid \mathbf{X}\right)$, 
		then
		\begin{align*}
			H( s;\mathbf{X})=F\left(k\mid \mathbf{X}\right)\text{ when } M_s^k\leq \mu<M_s^{k+1}.
		\end{align*}
		When $M_s^0=-\infty<\mu<M_s^1$, for example, $F\left(0\mid \mathbf{X}\right)$ is closest to $s$, and thus $H( s;\mathbf{X})=F\left(0\mid \mathbf{X}\right)$ by the definition of $H( s;\mathbf{X})$. 
		When $\mu=M_s^1$, $F(0\mid \mathbf{X})$ and $F\left(1\mid \mathbf{X}\right)$ are equidistant from $s$. While $M_s^1<\mu< M_s^2$, $F\left(1\mid \mathbf{X}\right)$ is closest to $s$ and thus $H( s;\mathbf{X})=F\left(1\mid \mathbf{X}\right)$. 

		The following Lemma \ref{finite} and Assumption \ref{monofnew}  are made to handle the non-smoothness issue for discrete outcomes with an infinite range. 
		Lemma \ref{finite} guarantees the summation on the left of \eqref{eq:densbound}  can be up to a large number $a_n$ going to $\infty$, and $f_{H( s;\mathbf{X})}(\cdot)$ can be then approximated by $\sum_{k=0}^{a_n}f_{F\left(k\mid \mathbf{X}\right)}(\cdot)$, which is smooth in the $\epsilon_n$-neighborhood of $s$. 
		
		\begin{lemma}\label{finite}
			There exists a sequence $a_n$ going to infinity such that, 
			for any $s\in[s_L,s_U]$,
			$f_{H( s;\mathbf{X})}(s+\epsilon_n)\geq \sum_{k=0}^{a_n}f_{F\left(k\mid \mathbf{X}\right)}(s+\epsilon_n)$.
		\end{lemma}
		
	The proofs of all the  theoretical results can be found in the supplementary material.
		The following assumption constrains the tail probability of $\mu$ to ensure $\sum_{k=0}^{a_n}f_{F\left(k\mid \mathbf{X}\right)}(\cdot)$ is a good approximation to $f_{H( s;\mathbf{X})}(\cdot)$.

		\begin{assumption}\label{monofnew}
			Let $a_n$ be the sequence  in Lemma \ref{finite}, then
			$\epsilon_n^{-2}P\left(\mu>M_s^{a_n}\right)\rightarrow0$, for any $s\in[s_L,s_U]$. 
		\end{assumption}

		We make the following regularity assumption to ensure $f_{F\left(k\mid \mathbf{X}\right)}$ is sufficiently smooth.
		\begin{assumption}\label{monof} 
			For fixed $k$, $f_{F\left(k\mid \mathbf{X}\right)}$ is twice continuously differentiable, and	$g$ and its derivatives $g'$ and $g''$  are uniformly bounded. In addition, for any $k$ and $k'$, the joint density of $\left(F\left(k\mid \mathbf{X}\right),F(k'\mid \mathbf{X})\right)$ is bounded.
		\end{assumption}
		
		A necessary yet not sufficient condition for Assumption \ref{monof} is that there exists at least one continuous  covariate whose coefficient is not 0. 
		When $Y$  follows a   Poisson distribution with mean $\lambda=\exp(\mu)$, 
		Assumptions  \ref{monofnew} and \ref{monof}  are satisfied if $\mathrm{E}_{\mathbf{X}}(\lambda)$ is finite, and they hold for negative binomial distributions 
		if $\mathrm{E}_{\mathbf{X}}(\lambda^2)$ is finite; see the supplementary material 
		for verification.
		Therefore, if there are highly right-skewed covariates, the log transformation is suggested. For binary and ordinal variables, Assumption \ref{monof}  is satisfied if the density of  $\mu$ is twice continuously differentiable.
		
		Now, we make the following assumptions in order to guarantee the convergence of the quasi-empirical residual distribution function when the estimated coefficients are plugged in. Denote $H(s; \mathbf{x}, \bm{\theta})$ as the closet interior grid point to $s$ when the parameters are set to be $\bm\theta$.

		\begin{assumption}[Lipschitz condition]\label{lips}
			There exists a constant $\alpha_2$  such that, for all for bounded $\bm\theta$ and $\bm{\theta}'$, when $\mid \bm\theta- \bm{\theta}'\mid $ is small enough, for any $s\in V$,
			$$\mid \mid H(s; \mathbf{x}, \bm\theta) -s\mid - \mid H(s; \mathbf{x},\bm{\theta}')-s\mid \mid \leq  \alpha_2\mid \bm\theta - \bm{\theta}'\mid. $$
		\end{assumption}
		This assumption is satisfied when $Y$ follows Poisson GLMs with a log link and bounded covariates.
		The following assumption guarantees the model estimation is well taken care of.
		\begin{assumption}\label{bigop}
			$\sqrt{n}(\hat{ {\bm\beta}}-\bm{\beta})=O_p(1)$.
		\end{assumption}
		\end{appendices}
	{\centering\section*{SUPPLEMENTARY MATERIAL}}
\begin{description}
	
	\item[Supplementary material:] The supplementary material includes additional simulation results and proofs of the theoretical results in Sections  \ref{asym} and Appendix \ref{sec:assume}. (.pdf file)
	
	\item[{\tt R} Code:] {\tt R}  implementation of the proposed method and several working examples to demonstrate its use. {\tt R}  code for downloading and analyzing the data.
	
\end{description}

\singlespacing
\bibliographystyle{econ}
\bibliography{residualbib}

\end{document}